\documentclass[11pt]{article}
\pdfoutput=1 

\usepackage{lmodern}

\usepackage[margin=1in]{geometry}
\usepackage{amsfonts,amssymb,amsmath,amsthm,mathtools,thmtools}
\usepackage{microtype,xspace,bm,floatrow}
\usepackage{braket}
\usepackage{algorithm}
\usepackage{algpseudocode}
\usepackage{comment}

\usepackage{doi}
\usepackage{hyperref}
\usepackage[svgnames]{xcolor}
\hypersetup{
	colorlinks=true,
	urlcolor=MidnightBlue,
	linkcolor=MidnightBlue,
	citecolor=MidnightBlue,
	unicode
}

\usepackage{enumitem}
\setlist{  
  listparindent=\parindent,
  parsep=0pt,
}

\usepackage{float}
\usepackage[capitalise,nameinlink,noabbrev]{cleveref}
\crefname{equation}{}{}

\usepackage{tikz}
\usepackage{graphicx}
\usetikzlibrary{positioning,arrows.meta,math,shapes.geometric,decorations.pathmorphing}

\usepackage[font=small,labelfont=bf]{caption} 
\usepackage{subcaption}

\usepackage[english]{babel}
\addto\extrasenglish{}
\addto\extrasenglish{}

\usepackage{lpic}

\newfloat{protocol}{htbp}{lop}
\floatname{protocol}{Protocol}

\crefname{protocol}{protocol}{protocols}
\Crefname{protocol}{Protocol}{Protocols}

\makeatletter
\DeclareRobustCommand\bfseries{%
  \not@math@alphabet\bfseries\mathbf
  \fontseries\bfdefault\selectfont\boldmath}
\makeatother

\newtheorem{theorem}{Theorem}[section]

\newtheorem{lemma}[theorem]{Lemma}

\newtheorem{claim}[theorem]{Claim}

\newtheorem{corollary}[theorem]{Corollary}
\newtheorem{definition}[theorem]{Definition}
\newtheorem{remark}[theorem]{Remark}

\newcommand{\newclass}[2]{\newcommand{#1}{{\text{\upshape\sffamily #2}}\xspace}}
\renewcommand{\P}{{\text{\upshape\sffamily P}}\xspace}
\newclass{\NP}{NP}
\newclass{\coNP}{coNP}
\newclass{\FP}{FP}
\newclass{\TFNP}{TFNP}
\newclass{\BPP}{BPP}
\newclass{\FBQP}{FBQP}
\newclass{\FBPP}{FBPP}
\newclass{\BQP}{BQP}
\newclass{\EQP}{EQP}
\newclass{\PH}{PH}

\newcommand{\newProbClass}[2]{\newcommand{#1}{{\text{\upshape\ttfamily #2}}\xspace}}
\newProbClass{\Prunable}{Prunable}
\newProbClass{\Fixable}{Fixable}
\newProbClass{\Garblable}{Garblable}

\newcommand{\newprob}[2]{\newcommand{#1}{{\text{\upshape\scshape #2}}\xspace}}
\newprob{\eol}{EoL}
\newprob{\findone}{Find1}
\newprob{\HHit}{HeavyHitter}
\newprob{\BalancedFindone}{BalancedFind1}
\newprob{\Grover}{Grover}
\newprob{\AOES}{AOES}
\newprob{\XOR}{XOR}
\newprob{\AOESfull}{Avoid-One-Encrypted-String}
\newprob{\ContSupp}{Continuous-Support-Selection}
\newprob{\UnifSuppFind}{Uniform-Support-Finding}
\newprob{\Hamming}{Hamming}
\newprob{\SimonHamming}{QL-Estimation}
\newprob{\distinctness}{Distinctness}
\newprob{\ksum}{Sum}
\newprob{\GraphCollision}{Graph-Collision}
\newprob{\HM}{Hidden-Matching}
\newprob{\Triangle}{Triangle}

\ifdefined\SHOWCOMMENTS
  \newcommand{\jiawei}[1]{{\color{brown}[Jiawei: #1]}}
  \newcommand{\hugo}[1]{{\color{blue}[Hugo: #1]}}
  \newcommand{\tom}[1]{{\color{purple}[Tom: #1]}}
\else
  \newcommand{\jiawei}[1]{{}}
  \newcommand{\hugo}[1]{{}}
  \newcommand{\tom}[1]{{}}
\fi

\def\NOTANONYMOUS{1}

\newcommand{\XORfunc}[1]{\mathsf{XOR}(#1)}

\newcommand{\alg}{\mathcal{A}}
\newcommand{\probR}{{R}}
\newcommand{\vL}{\mathit{L}}
\newcommand{\vR}{\mathit{R}}
\newcommand{\vcur}{\mathit{Cur}}
\newcommand{\calX}{\mathcal{X}}
\newcommand{\calY}{\mathcal{Y}}
\newcommand{\rvX}{\mathbf{X}}
\newcommand{\rvY}{\mathbf{Y}}
\newcommand{\setS}{\mathcal{F}}
\newcommand{\bv}[1]{{#1}}
\newcommand{\Flip}{\mathcal{E}}
\newcommand{\Fliplow}{\mathcal{E}_{\mathsf{low}}}
\newcommand{\Xh}{X^*}
\newcommand{\Ylow}{Y_{\mathsf{low}}}
\newcommand{\Yhigh}{Y_{\mathsf{hi}}}

\newcommand{\D}{\mathsf{D}}
\newcommand{\R}{\mathsf{R}}
\newcommand{\Q}{\mathsf{Q}}
\newcommand{\psR}{\mathsf{psR}}
\newcommand{\psQ}{\mathsf{psQ}}
\newcommand{\V}{\mathsf{V}}
\newcommand{\concat}[2]{#1 \circ #2}
\newcommand{\QC}{Q}

\newcommand{\LB}{Next}

\begin{document}

\mbox{}\vspace{12mm}

\begin{center}
{\huge Pseudo-Deterministic Quantum Algorithms}
\\[1cm] \large

\ifdefined\NOTANONYMOUS
    \setlength\tabcolsep{2em}
    \begin{tabular}{ccc}
    Hugo Aaronson & Tom Gur & Jiawei Li
    \\[-0mm]
    \small\slshape University of Cambridge & \small\slshape University of Cambridge & \small\slshape UT Austin 
    \end{tabular}
    	
    \vspace{6mm}
\else

\fi
	
\large
\today
	
\vspace{6mm}
	
\end{center}

\begin{abstract}
    We initiate a systematic study of pseudo-deterministic quantum algorithms. These are quantum algorithms that, for any input, output a canonical solution with high probability. Focusing on the query complexity model, our main contributions include the following complexity separations, which require new lower bound techniques specifically tailored to pseudo-determinism:

\begin{itemize}
    \item We exhibit a problem, {Avoid One Encrypted String} (\AOES), whose classical randomized query complexity is $O(1)$ but is maximally hard for pseudo-deterministic quantum algorithms ($\Omega(N)$ query complexity).
    \item We exhibit a problem, {Quantum-Locked Estimation} (\SimonHamming),
    for which pseudo-deterministic quantum algorithms admit an exponential speed-up over classical pseudo-deterministic algorithms ($O(\log(N))$ vs. $\Theta(\sqrt{N})$), while the randomized query complexity is $O(1)$. 
\end{itemize}

    Complementing these separations, we show that for any total problem $R$, pseudo-deterministic quantum algorithms admit at most a quintic advantage over deterministic algorithms, i.e., $\D(R) = \tilde O(\psQ(R)^5)$.
    On the algorithmic side, we identify a class of quantum search problems that can be made pseudo-deterministic with small overhead, including Grover search, element distinctness, triangle finding, $k$-sum, and graph collision. 
\end{abstract}

\section{Introduction}

The study of pseudo-deterministic algorithms, initiated by Gat and Goldwasser~\cite{GG11}, is concerned with randomized algorithms that output a canonical solution for each input. That is, a randomized algorithm $\mathcal{A}$ that solves a search problem $\probR\subseteq\mathcal{X}\times \mathcal{Y}$ is said to be pseudo-deterministic if there exists a function $f:\mathcal{X}\rightarrow\mathcal{Y}$ such that for any input $X\in\mathcal{X}$, the output $\mathcal{A}(X)$ is $f(X)$ with high probability. 

Pseudo-deterministic algorithms capitalize on the computational power of randomization while still maintaining the predictable output behaviour of deterministic algorithms. This stands in contrast to fully randomized algorithms, whose repeated runs may return vastly different outputs.
Over the last decade, pseudo-deterministic algorithms have received much attention and have appeared across diverse areas of theoretical computer science~\cite{GGR13, GGH18,GGH19,GGMW20,DPV21,DPWV22,CDM23query, CLORS23, BKKS23,GGS23}.

In the quantum setting, pseudo-deterministic algorithms were first considered by Goldwasser, Impagliazzo, Pitassi and Santhanam~\cite{GIPS21} when studying the \findone problem\footnote{The input to the \findone problem is a string $X\in\{0,1\}^N$ with the promise that $X$ has Hamming weight at least $N/2$, and the goal is to output an index $i$ with $X_i=1$.} 
for pseudo-deterministic algorithms in the query complexity model. While \findone is solvable in $O(1)$ queries by random sampling, they showed that any pseudo-deterministic algorithm, classical or quantum, for \findone requires ${\Omega}(\sqrt{N})$ queries. 

What makes quantum pseudo-determinism truly distinct is the randomness that arises from \emph{measurement}.
Many quantum algorithms that achieve exponential quantum speed-up (e.g., Fourier sampling) measure superpositions over exponentially many candidate solutions, a step that appears inherent.  This offers a new perspective to revisit existing quantum advantages: 

\begin{quote}
    \begin{center}
    \emph{When can quantum speed-ups be achieved pseudo-deterministically,\\ and when does quantum advantage inherently require output entropy?}
    \end{center}
\end{quote}

This question also links quantum pseudo-determinism to certified randomness: if one can show that the quantum advantages of a certain problem $\probR$ cannot be preserved given the pseudo-deterministic requirement, then any (potentially untrusted) quantum device $\alg$ that solves $\probR$ efficiently must generate some entropy.
A series of works~\cite{bassirian2021certified, AH23, Yamakawa2022} show, under standard complexity assumptions, that Fourier sampling~\cite{fishing}, random-circuit sampling, and the Yamakawa–Zhandry algorithm~\cite{Yamakawa2022} cannot be made pseudo-deterministic\footnote{For quantum certified randomness, these works in fact prove {quantitatively} stronger statements: any efficient quantum algorithm for these tasks must generate a large amount of min-entropy.}. 

Motivated by the above, we initiate a systematic study of quantum pseudo-determinism. We investigate the fundamental properties of pseudo-deterministic quantum algorithms, show complexity separations that shed light on the strength and limitations of quantum pseudo-determinism, and identify a wide class of problems (which includes Grover search, element distinctness, triangle finding, $k$-sum, and graph collision) that admit efficient pseudo-deterministic quantum algorithms.

\subsection{Our results}

We start with the query complexity model. Unless otherwise stated, we denote the length of the input to the search problem by $N$. For a search problem $\probR$, we denote the deterministic query complexity by $\D(\probR)$, the pseudo-deterministic quantum query complexity $\psQ(\probR)$, and the quantum query complexity by $\Q(\probR)$. We denote the output of algorithm $\alg$ given query access to input $X$ as $\mathcal{A}^{X}$.

\subsubsection{Query complexity separations}
We first study the relationship between pseudo-deterministic quantum query algorithms, randomized algorithms, and classical pseudo-deterministic algorithms. 

\paragraph{Hardness of quantum pseudo-determinism.} Our first result shows that not all problems admit efficient pseudo-deterministic quantum algorithms. We give a maximal separation between classical randomized query complexity and quantum pseudo-deterministic query complexity. 

\begin{restatable}{theorem}{IntroAOES}\label{thm:IntroAOES}
    There exists a search problem $\probR$ whose classical randomized complexity is $O(1)$ and whose pseudo-deterministic quantum query complexity is $\Omega(N)$.
\end{restatable}

\autoref{thm:IntroAOES} improves the separation given by the \findone problem, where the pseudo-deterministic quantum lower bound is $\Omega(\sqrt{N})$. 
We obtain this stronger separation via the \emph{Avoid One Encrypted String} (\AOES) problem. Informally, the input encodes $m$ instances of \XOR gadgets; together they encrypt a hidden string $b\in\{0,1\}^m$, where $b_i$ is the solution to the $i$-th \XOR instance. The goal is to output any $m$-bit string \emph{not equal} to $b$. A classical randomized algorithm succeeds with constant probability by blind sampling, whereas any pseudo-deterministic algorithm must consistently avoid $b$ across runs, effectively forcing it to learn at least one bit of $b$.

\paragraph{Power of quantum pseudo-determinism.}
Our second result shows that pseudo-deterministic quantum algorithms can be much stronger than their classical counterparts. We prove an exponential separation between pseudo-deterministic quantum algorithms and pseudo-deterministic classical algorithms on a task that remains easy for classical randomized algorithms. To ensure this separation is meaningful, we need to be certain that the quantum algorithm is not simply introducing solutions inaccessible to classical computation. We do this by defining the notion of ``canonization''. Intuitively, this is where a quantum algorithm stabilizes a solution that occurs with considerable probability in a randomized algorithm into a canonical solution, thus canonizing it from a randomized solution to a pseudo-deterministic one.

\begin{definition}\label{def: Elevate}
    Let $\probR\subseteq\mathcal{X}\times \mathcal{Y}$ be a search problem with randomized algorithm $\mathcal{A}$ and let $f:\mathcal{X}\rightarrow\mathcal{Y}$ be a function. $f$ canonizes $\mathcal{A}$ if, for any $X\in \mathcal{X}$, the probability $\mathcal{A}^{X}$ outputs $f(X)$ is greater than the probability it outputs an incorrect solution. Equivalently,
    \begin{equation*}
    \forall X\in \mathcal{X}:\mathbb{P}_{\mathcal{A}}[\mathcal{A}^X=f(X)]> \mathbb{P}_{\mathcal{A}}[ (X,\mathcal{A}^X)\not\in \probR].
\end{equation*}
A pseudo-deterministic algorithm $\mathcal{B}$ canonizes $\mathcal{A}$ if the function $f$ calculated by $B$ canonizes $\mathcal{A}$.
\end{definition}

We now state the theorem that demonstrates the quantum advantage on pseudo-deterministic algorithms.

\begin{restatable}{theorem}{IntroSH}\label{thm:IntroSH}
    There exists a search problem $\probR$ such that
    \begin{enumerate}
        \item there exists a randomized classical algorithm $\alg$ for $\probR$ that uses $O(\log N)$ queries; 
        \item there exists a quantum pseudo-deterministic algorithm $\mathcal{B}$ with $O(\log N)$ query complexity that canonizes $\alg$;
        \item any classical pseudo-deterministic algorithm for $\probR$ requires $\Omega(\sqrt{N})$ queries.
    \end{enumerate}
\end{restatable}

We make two remarks about the conditions required for showing a conceptually meaningful quantum advantage in the setting of pseudo-deterministic algorithms. First, note that without requiring the problem to be easy for randomized algorithms (Item~1), such a separation would be trivial, in the sense that it need not refer to pseudo-determinism at all (e.g., by Simon’s problem, which has a unique solution). In contrast, \autoref{thm:IntroSH} exhibits a different form of quantum advantage: beyond faster search, quantum algorithms can \emph{break symmetry} (by eliminating output randomness) more efficiently than classical pseudo-deterministic algorithms.
Furthermore, we also require that such a quantum pseudo-deterministic algorithm canonizes a classical one. This is because we do not want the quantum algorithm to \emph{cheat} by outputting a new solution that is inaccessible to classical algorithms. 

\subsubsection{Fundamental properties of quantum pseudo-determinism.}
Next, we provide theorems that demonstrate two basic properties of pseudo-deterministic quantum algorithms: a limitation on maximal speed-ups for total problems, and a general pseudo-deterministic quantum query upper bound via the completeness of the \findone{} problem.

\paragraph{Deterministic emulation of total problem.} We show that for total problems, pseudo-deterministic quantum algorithms cannot exhibit exponential speed-up over classical algorithms. In particular, pseudo-deterministic quantum algorithms for total problems can be emulated by deterministic algorithms with a quintic blow-up. 

\begin{restatable}{theorem}{TotalPoly}\label{thm:TotalPoly}
    Let $\probR \subseteq \{0,1\}^{N}\times \calY$ be a total search problem,
    then $\D(R)=O(\psQ(R)^{5}\log(N))$.
\end{restatable}

In contrast, in the classical setting, it was shown in \cite{GGR13} that pseudo-deterministic classical algorithms for total problems can be emulated by deterministic algorithms with a quartic blow-up $\D(R)=O(\psR(R)^{4}\log(N))$. Our proof of \cref{thm:TotalPoly} follows the strategy of \cite{GGR13}, but poses a new challenge; namely, our main technical contribution in the proof of \cref{{thm:TotalPoly}} is a lemma (\cref{lem:MaxOutput}) that upper bounds the number of possible outputs of a pseudo-deterministic quantum algorithm; this step is trivial in the classical case.

\paragraph{General pseudo-deterministic quantum query upper bound.}
We also observe that the completeness result in \cite{GIPS21}, which states that the \findone problem is \emph{complete} for pseudo-deterministic classical query complexity, extends to the quantum setting. By combining it with the fact that $\psQ(\findone)=\Theta(\sqrt{N})$~\cite{GIPS21},\footnote{An $\tilde{O}(\sqrt{N})$ upper bound for \findone is informally given in \cite{GIPS21}. We make a more careful analysis and show that this logarithmic term can be avoided. See \cref{sec:GrovWit} and \cref{thm:PsDGrov} for details.} we derive a general upper bound on the pseudo-deterministic quantum query complexity of any problem via its randomized query complexity and verification complexity (i.e., the deterministic query complexity of verifying whether an input-solution pair is in a relation).

\begin{restatable}{theorem}{FindoneQComplete}\label{thm:Find1QComplete} 
For any search problem $\probR$ with randomized query complexity $\R(R)$ and verification complexity $\V(R)$, we have $\psQ(R) = O\big((\R(R)+\V(R))\cdot \sqrt{N}\big)$.
\end{restatable}

This implies that any lower bound for a pseudo-deterministic quantum query complexity problem greater than $\Omega(\sqrt{N})$ requires either a high randomized query complexity or it is difficult to verify. For instance, in our maximally pseudo-deterministically hard problem \AOES (\cref{def: AOES}), the randomized query complexity is $O(1)$, but it is very hard to verify.

\subsubsection{Making quantum algorithms pseudo-deterministic.}
In \cite{GIPS21}, it was observed that by integrating Grover’s algorithm in a binary search, one can obtain a pseudo-deterministic quantum algorithm for \findone that returns the first non-zero index in $\tilde{\Theta}(\sqrt{N})$ queries.

We extend this idea further and generalize it to obtain pseudo-deterministic quantum algorithms for a broad family of \emph{$k$-subset finding} problems~\cite{CE05subset}. Here the input $X\in\Sigma^N$ consists of $N$ elements, and the goal is to find a subset of $k$ elements satisfying a predicate. To make the binary-search approach applicable, we introduce a mild condition, \Prunable, which informally requires that restricting the search domain does not increase the query cost. We show that for \Prunable{} problems, any subset-finding quantum algorithm can be turned pseudo-deterministic with only a near-linear dependence on $k$.

\begin{restatable}{theorem}{prunable}\label{thm:prunable}
    Let $\probR$ be a $k$-subset finding problem. If $\probR$ is \Prunable, then $\psQ(\probR) = \tilde{O}(k\cdot \Q(\probR))$.
\end{restatable}

We note that many natural $k$-subset finding tasks are \Prunable, e.g., $k$-distinctness, triangle finding, and $k$-sum. However, we observe that the condition can fail for $k$-subset finding problems with global structure, such as \findone{}, when interpreted as a $1$-subset finding problem.

\subsubsection{Beyond the query complexity model}
Finally, we study pseudo-deterministic quantum algorithms outside the setting of the standard query complexity model; namely, by considering search problems with quantum inputs, as well as the white-box setting.

\paragraph{Search problems with quantum inputs.}

We consider a different sort of search problems for which the input is a quantum state. Here, the query complexity is defined as the number of copies of the input state required to solve the problem.
Pseudo-determinism could be useful in problems such as finding a computational basis element in the support of the input state, finding a state orthogonal to the input state, etc. We show, however, there can be no pseudo-deterministic algorithms if the input domain is a connected continuous space.  

Nevertheless, we consider a quantum version of \findone we call \UnifSuppFind: finding a computational basis element in the support of an input state $\ket{\psi}$, where $\ket{\psi}$ is promised to be a \emph{uniform} superposition over a subset of basis. We prove that the query complexity of \UnifSuppFind is $\Theta(d)$, where $d$ is the dimension of $\ket{\psi}$ (\autoref{thm: suppfind}).

\paragraph{The white-box setting.}\label{sec:WB}

Finally, we give the characterization of problems that admit efficient pseudo-deterministic quantum algorithms in the white-box (Turing machine) setting.

\begin{restatable}{theorem}{WhiteBox}\label{thm: wb PsQ to FPBQP}
    A search problem $\probR$ has a pseudo-deterministic quantum polynomial-time algorithm if and only if $\probR$ is $\P$-reducible\footnote{A $\P$-reduction refers to a deterministic polynomial time Turing reduction.} to some decision problem in $\BQP$.
\end{restatable}

We prove \autoref{thm: wb PsQ to FPBQP} in \cref{App:WhiteBox}, which is generalized from its classical analog in \cite{GGR13}. Intuitively, \autoref{thm: wb PsQ to FPBQP} says that any search problem $R$ that has an efficient pseudo-deterministic quantum algorithm if and only if it has a \emph{search-to-decision} reduction to $\BQP$. In other words, any (super-polynomial) pseudo-deterministic quantum advantage for a search problem essentially comes from the quantum advantages for a certain decision problem.

\subsection{Techniques}

In this section, we provide a taste of our techniques for establishing the query complexity separations in \autoref{thm:IntroAOES} and \autoref{thm:IntroSH}. In both cases, the key part is proving query complexity lower bounds for pseudo-deterministic quantum or classical algorithms, where new techniques are needed. This can be attributed to a lack of tools for proving pseudo-deterministic lower bounds, and that well-established techniques such as the polynomial method cannot be directly applied.

\paragraph{A maximal separation between randomized and pseudo-deterministic quantum algorithms.}

A quintessential advantage of randomized algorithms compared to pseudo-deterministic ones is their ability to output random strings.
We consider the following problem, called \emph{Avoid One Encrypted String} (\AOES), which capitalizes on this advantage. The input of \AOES consists of $m$ instances of the \XOR problem in parallel, where $m$ is a big constant. Let $N$ be the total input length and $N/m$ be the input length of each \XOR instance. These \XOR instances encode a secret string $b \in \{0,1\}^m$, where $b_i$ is the value of the $i$-th \XOR instance. The goal is to output any $m$-bit string that is not equal to $b$.

\AOES can be solved trivially by randomly sampling an $m$-bit string with accuracy $1 - 1/2^m$. However, it is not difficult to show that no deterministic algorithm can do better than using $N/m$ queries to solve the first \XOR instance, and then output any string with the first bit unmatched. 
Intuitively, any pseudo-deterministic algorithm that can avoid outputting $b$ consistently must have learned some information about those \XOR instances. On the other hand, \XOR is maximally hard for quantum query algorithm, i.e., any quantum query algorithm for \XOR of length $n$ cannot do any better than a random guess with less than $n/2$ queries~\cite{farhi1998limit, BBCMW01}. 

To turn the intuitions to a proof, we provide a new \emph{randomized} reduction from \XOR to \AOES. Namely, we show that if there exists a zero-error quantum algorithm $\mathcal{B}$ for $\AOES_{N,m}$ with query complexity $q$, then there exists a quantum algorithm $\mathcal{A}$ for $\XOR_{N/m}$ which calls $\mathcal{B}$ once such that for every input $ X\in \{0,1\}^{N/m}$,
    \begin{equation*}
        \mathbb{P}\left[\mathcal{A}^{X}=\XORfunc{X} \right]\geq \frac{1}{2} +\frac{1}{2^{m+1}}.
    \end{equation*}

To prove that, let $\mathcal{B}$ be a solver for the \AOES problem, and let $X \in \{0,1\}^{N/m}$ be the \XOR instance we want to solve. We build a random \AOES instance $Z = (Z_1, \ldots, Z_m)$ in two steps. First, we generate $m$ uniformly random string $Y_1, \ldots, Y_m \in \{0,1\}^{N/m}$. Then, for each position $i\in [m]$, we take $Z_i = Y_i$ with probability $1/2$, and $Z_i = Y_i \oplus X$ with probability $1/2$. 

At a high-level, $Y_1, \ldots, Y_m$ can be viewed as a one-time pad, and $X$ is embedded into a random subset of positions in the \AOES instance $Z$. The crucial observation is that from the \AOES solver $\mathcal{B}$'s perspective, $Z$ is uniformly random, and it has no information about $X$, and where $X$ is embedded in $Z$.
We can show that when we embed $X$ into a certain subset of positions in $Z$ that depends on $\mathcal{B}$, we can then recover the solution of the \XOR instance $X$ with certainty. When $\mathcal{B}$ is a zero-error algorithm for \AOES, the overall success rate of our randomized reduction is $1/2 + 1/2^{m+1}$, and any value strictly greater than $1/2$ suffices given the hardness of \XOR. 

Finally, for any pseudo-deterministic algorithm $\mathcal{B}$ for \AOES, we can decrease its error rate to $1/2^{2m}$ via $O(m)$ repetitions and picking the most common outputs. We can then plug this extremely low error \AOES algorithm into the randomized reduction and still get non-zero advantages on \XOR. Note that while the random sampling algorithm has an error rate of $1/2^m$, it cannot further reduce it via repetition.

\paragraph{Exponential pseudo-deterministic quantum advantage.}

We construct a problem called Quantum-Locked Estimation (\SimonHamming), which is easy for randomized algorithms as well as quantum pseudo-deterministic algorithms, but is hard for classical pseudo-deterministic algorithms.

Let $N = 2^n$. For simplicity, we sometimes interpret a value $x \in [N]$ as an $n$-bit string in the natural way, and vice versa.
Recall that in Simon's problem~\cite{Simon97}, 
the input is a function $f:[N]\rightarrow[N]$ which is promised to encode a secret string $s$ such that $f(x)=f(y)\iff (x=y) \vee (x=y\oplus s)$. In the \Hamming problem, an input string $X\in \{0,1\}^N$ is given, and the goal is to output an estimate of the Hamming weight of $X$ up to additive error $N/10$. 

In the \SimonHamming problem, the input is simply a pair $(f,X)$, where $f$ is an instance of the Simon's problem and $X$ is an instance of the \Hamming problem.
It is promised that $f$ encodes some string $s$ that is promised to be a valid estimate of the Hamming weight of $X$ ($\big|s-\|X\|_{\mathsf{hw}}\big|\leq 0.09\cdot N$). 
The goal is to output an estimate for the Hamming weight of $X$ with an additive error up to $N/10$.

The classical randomized algorithm can estimate the Hamming weight with $O(1)$ queries, and a quantum algorithm can calculate the secret string $s$, which is a valid solution by the promise of the problem. Indeed, this quantum algorithm \emph{canonizes} the particular estimate $s$, which is one of the many different estimates of the Hamming weight that a classical randomized algorithm could reasonably produce.\footnote{To ensure the condition of the canonization (\cref{def: Elevate}) is satisfied, we need a slightly adjusted classical random sampling algorithm with $O(\log N)$ queries. See \cref{sec:SimonHamming} for more details.}
We view this as a ``lock-and-key'' construction, where 
the \Hamming problem locks a canonical solution that is hidden from pseudo-deterministic classical algorithms, while the key is Simon's problem, to which only quantum algorithms have access.

For the lower bound of pseudo-deterministic classical algorithms, 
it is helpful to first revisit the pseudo-deterministic hardness of \Hamming via a \emph{sensitivity argument}. Suppose for some string $X$ (where $\|X\|_{\mathsf{hw}}=N/2$), the pseudo-deterministic output of some algorithm is $t$. Let $X'$ be the string with the smallest Hamming weight whose pseudo-deterministic output is $t$. This algorithm is  \emph{sensitive} at $X'$ as flipping any $1$-bit of $X'$ will result in a string whose pseudo-deterministic output is not $t$ (as it will have a smaller Hamming weight), which implies a large query complexity of differentiating $X'$ from any such string. 

However, a high sensitivity also implies a large quantum query complexity, therefore, we have to modify our argument for the \SimonHamming problem, which is easy for a quantum algorithm. Suppose $\alg$ is a classical pseudo-deterministic algorithm for \SimonHamming. The crux of the proof is arguing that if $\alg$ has sub-linear sensitivity everywhere, then it must \emph{approximately} solve the secret string $s$ from the Simon's problem (\cref{lem: reduce Simon to SH}). 
To this end, we apply a similar trick from \Hamming's proof: finding a string $X'$ with the smallest Hamming weight among those strings that yield the same output by $\alg$; then flipping its $1$-bits. The analysis here is more delicate because the Hamming and Simon components of the problem are not independent; the Simon witness is coupled with the Hamming input, and therefore standard product-distribution or decoupling arguments do not apply.

Finally, we show that solving Simon's problem approximately is as hard as solving it exactly, and hence we conclude that $\alg$ must use $\Omega(\sqrt{N})$ queries.

\subsection{Paper Organization}

In \cref{sec:prelim} we set up the query model, define pseudo-deterministic classical/quantum algorithms, and review key preliminaries (including Find1 and information-theoretic tools). \cref{sec:Basic} proves foundational results, including the quintic relationship for total search problems and the completeness of Find1 for quantum pseudo-determinism. \cref{sec:Sep} gives our main separations: \cref{sec:AOES} introduces AOES to obtain a maximal gap between randomized and pseudo-deterministic quantum query complexity, and \cref{sec:SimonHamming} defines \SimonHamming to separate pseudo-deterministic quantum from pseudo-deterministic classical algorithms while remaining randomized-easy. \cref{sec:GrovWit} shows how to make a broad class of quantum subset-finding algorithms pseudo-deterministic with small overhead. Finally, \cref{sec:quantInputs} studies pseudo-determinism for search problems with quantum inputs, proving topological limitations and tight bounds for a quantum Find1 analogue.

\section{Preliminaries}\label{sec:prelim}

In the query complexity model, an algorithm $\mathcal{A}$ solves a search problem $\probR \subseteq \mathcal{X}\times \mathcal{Y}$ if, given query access to the input $X \in \mathcal{X}$, the algorithm outputs $Y\in \mathcal{Y}$ such that $(X,Y)\in \probR$ with probability at least $2/3$. We denote the output of algorithm $\alg$ given query access to input $X$ as $\mathcal{A}^{X}$.
Unless otherwise stated the set of inputs is a subset of the binary strings, $\mathcal{X}\subseteq \bigcup_{N\in\mathbb{N}}\{0,1\}^{N}$. In this case, the algorithm has explicit access to the input length $N\in\mathbb{N}$ for which $X\in \{0,1\}^N$. The complexity measure here is the number of queries made to the black box oracle. Unless there are other parameters to the search problem we will define $\mathcal{X}_{N}\coloneqq \mathcal{X}\cap \{0,1\}^N$ to be the subset of the search problem for inputs of length $N$. When the length of the input is clear from context we will suppress the subscript and write only $\mathcal{X}$.

\begin{definition}[Pseudo-deterministic algorithm]\label{def:PsD_alg}
    An algorithm $\mathcal{A}$ for a search problem $\probR\subseteq\mathcal{X}\times \mathcal{Y}$ is said to be pseudo-deterministic if there exists a function $f:\mathcal{X}\rightarrow\mathcal{Y}$ such that for all $X\in\mathcal{X}$, the output of $\mathcal{A}(X)$ is $f(X)$ with probability at least $2/3$ and the output of $f$ is always a solution to the search problem. Here, the output $f(X)$ is said to be the canonical or pseudo-deterministic solution. Equivalently,
    \begin{equation*}
        \exists f:\mathcal{X}\rightarrow\mathcal{Y},\forall X\in \mathcal{X} \left((X,f(X))\in \probR \wedge \mathbb{P}[\mathcal{A}(X)=f(X)]\geq \frac{2}{3}\right).
    \end{equation*}
\end{definition}

In the classical setting a query to the black box is specified by the algorithm as an index $i\in[N]$ and then the algorithm receives the value $X_i$. For a quantum algorithm with a classical input, the algorithm is a circuit composed of unitaries and a single query is a single use of the oracle unitary $U_X$ whose action is expressed as
\begin{equation*}
    \forall i\in [N],j\in \{0,1\}:U_{X}\ket{i}\otimes\ket{j}=\ket{i}\otimes\ket{j\oplus X_{i}}.
\end{equation*}

\begin{definition}[Query complexity measures]
    For a search problem $\probR\subseteq \calX\times \calY$, we denote the following measures of complexity in the query setting.

    \begin{itemize}
        \item $\D(\probR)$ is the deterministic query complexity of solving $\probR$.
        \item $\R(\probR)$ is the randomized classical query complexity of solving $\probR$.
        \item $\Q(\probR)$ is the randomized bounded-error quantum query complexity of solving $\probR$.
        \item $\psR(R)$ is the pseudo-deterministic bounded error classical query complexity of solving $\probR$.
        \item $\psQ(\probR)$ is the pseudo-deterministic bounded-error quantum query complexity of solving $\probR$.
    \end{itemize}
    Here, bounded-error means that any algorithm must succeed with probability at least $2/3$. We define $\D,\R,\Q$ similarly for a Boolean function $f: \calX\rightarrow \{0,1\}$.
\end{definition}

\subsection{\textsc{Find1} problem}

We define the \findone problem and review the existing completeness and quantum lower bound results from \cite{GIPS21}.

\begin{definition}[\findone]
    On input $X\in \{0,1\}^{N}$ such that $\sum_{j}X_{j}\geq N/2$, a solution to the \findone problem is an index $i$ such that $X_{i}=1$. 
\end{definition}

The next theorem establishes the completeness of the \findone problem for pseudo-deterministic query algorithms.

\begin{theorem}[\cite{GIPS21} Theorem 3, Completeness of \findone]
    Let $r,q,v$ be functions satisfying $f(\Theta(N))=\Theta(f(N))$ for any $f\in \{r,q,v\}$
    Let $\probR$ be a search problem such that with input size $N$, a solution is verifiable with $v(N)$ queries, the bounded error randomized query complexity is at most $r(N)$, and the bounded error pseudo-deterministic query complexity of $\probR$ is at least $q(N)$. Then the pseudo-deterministic query complexity of \findone is at least $\Omega(q(N)/(r(N)+v(N)))$.
\end{theorem}

We shall also need the following lower bounds on the pseudo-deterministic quantum complexity of the \findone problem.
\begin{theorem}[\cite{GIPS21} Theorem 4, quantum lower bound]\label{thm:find1_lb}
    The pseudo-deterministic quantum query complexity of \findone is at least $\Omega(\sqrt{N})$.
\end{theorem}

\subsection{Entropy}

For a random variable $\rvX$, let $|\rvX|$ denote the size of the support of $\rvX$. Throughout, all logarithms are base $2$.

\begin{definition}[Shannon Entropy]
    For a random variable $\rvX$ with probability distribution $(p_{x})_{x}$, the Shannon entropy of $\rvX$ is defined as 
    \begin{equation*}
        H(\rvX)=-\sum_{x}p_{x}\log(p_{x}).
    \end{equation*}
\end{definition}

\begin{definition}[von Neumann Entropy]
    Let $\rho$ be a mixed quantum state with a set of eigenvalues $(p_{x})_{x}$, the von Neumann entropy of $\rho$ is defined as 
    \begin{equation*}
        S(\rho)=H(p)=-\sum_{x}p_{x}\log(p_{x}).
    \end{equation*}
\end{definition}

\begin{remark}\label{rem:maxEnt}
    Note that for both a distribution supported over $n$ elements and a quantum state of dimension $n$, the maximum value of the corresponding notion of entropy is $\log(n)$. This is achieved by the uniform distribution and the maximally mixed state respectively.
\end{remark}

We denote the joint entropy of two random variables $\rvX, \rvY$ by $H(\rvX, \rvY)$. We denote the conditional entropy by $H(\rvX\vert\rvY)=H(\rvX,\rvY)-H(\rvY)$.

\begin{definition}
    The mutual information between two random variables $\rvX,\rvY$ is $I(\rvX,\rvY)=H(\rvX)+H(\rvY)-H(\rvX,\rvY)=H(\rvX)-H(\rvX\vert\rvY)$. 
\end{definition}

\begin{lemma}[Holevo's Bound, Theorem 12.1 \cite{NC16}]\label{lem:Holevo}
    For an ensemble of quantum states $(\rho_{x})_{x\in [n]}$ and a probability distribution $\rvX=(p_{x})_{x\in[n]}$ and a POVM $\{L_{y}\}_{y\in [m]}$. If the output distribution of the POVM on the state $\rho=\sum_{x}p_{x}\rho_{x}$, is $\rvY=\{\mathrm{Tr}(L_{y}\rho)\}_{y\in[m]}$ The following holds, 
    \begin{equation*}
        I(\rvX,\rvY)\leq S(\rho)-\sum_{p_{x}}S(\rho_{x}).
    \end{equation*}
\end{lemma}

\begin{lemma}[Fano's Inequality, Box 12.2 \cite{NC16}]\label{lem:Fano}
    The conditional entropy between two classical random variables $\rvX,\rvY:H(\rvX|\rvY)=H(\rvX,\rvY)-H(\rvY)$ is upper bounded in the following way,
    \begin{equation*}
        H(\rvX|\rvY)\leq H(e,1-e)+e\log(|\rvX|-1).
    \end{equation*}
    Here, $e$ is the minimum error of any algorithm that estimates the value of $\rvX$ from $\rvY$.
\end{lemma}

\section{Fundamentals of pseudo-deterministic quantum query complexity}\label{sec:Basic}

In this section, we establish two fundamental properties of pseudo-deterministic quantum algorithms: the relation between deterministic and quantum pseudo-deterministic algorithms for total search problems, and the completeness of the Find1 problem for quantum pseudo-determinism.

\subsection{Maximum quintic speedup for total search problems}
We show that for any total search problem, pseudo-deterministic quantum algorithms offer no super-polynomial speed-up over classical algorithms as stated in \cref{thm:TotalPoly} that we restate below. 

\TotalPoly*

This result is a generalization of the classical analogue in \cite[Theorem 4.1]{GGR13}, which proved that there is a \emph{quartic} relationship between the classical pseudo-deterministic query complexity and deterministic query complexity of any total search problem. This was proved by first showing that the number of pseudo-deterministic outputs is upper bounded by $2^{q(N)+1}$, then calculating each of the $q(N)+1$ bits of the canonical output separately using the pseudo-deterministic algorithm. Each of these bits --- a total boolean function --- can be calculated deterministically with a polynomial overhead given the relationship between the deterministic query complexity and the randomized query complexity for total boolean functions.

We use a similar strategy for proving \autoref{thm:TotalPoly}. We first need an upper bound on the number of possible outputs of a pseudo-deterministic quantum algorithm, after this we can encode the canonical output in a bounded number of bits.
This step turns out to be a more challenging task than the classical case. Using tools from classical and quantum information theory, we show the following lemma, which may be of independent interest. 

\begin{lemma}\label{lem:MaxOutput}
    Let $\alg$ be a pseudo-deterministic quantum query algorithm for a search problem $R\subseteq \{0,1\}^{N}\times\{0,1\}^{\ell}$ with query complexity $q$. Let $f_{\mathcal{A}}:\{0,1\}^{N}\rightarrow\{0,1\}^{\ell}$
    denote the function that maps an element of $\{0,1\}^{N}$ to the canonical output of $\alg$.
    Then, the number of possible outputs, $|\mathrm{Im}(f_{\mathcal{A}})|$, is $O((2N)^{3q/2})$.
\end{lemma}
\begin{proof}
    Let $m=|\mathrm{Im}(f_{\mathcal{A}})|$
    and suppose the input is $X\in \{0,1\}^N$, and for all $i\in [q]$, registers $B_{i}$ and $C_{i}$ represent the input and the output to the oracle for $X$. Recall that the quantum oracle is expressed as a unitary $U_X$ with the following action,
    \begin{equation*}
        U_X:\ket{a}\otimes \ket{b}\mapsto \ket{a}\otimes \ket{b\oplus X_a}.
    \end{equation*} 
    The unitaries $(V_{i})_{i\in[q]}$ prepare the input to the oracle in register $B_{i}$ while leaving the previous registers unaffected,
    \begin{align*}
        V_{i}:&\ket{j_{1}}_{B_{1}}\otimes \ket{X_{j_{1}}}_{C_{1}}\otimes\ldots\otimes \ket{j_{i-1}}_{B_{i-1}}\otimes\ket{X_{j_{i-1}}}_{C_{i-1}}\otimes \ket{0}\mapsto \\
        &\ket{j_{1}}_{B_{1}}\otimes \ket{X_{j_{1}}}_{C_{1}}\otimes\ldots\otimes \ket{j_{i-1}}_{B_{i-1}}\otimes\ket{X_{j_{i-1}}}_{C_{i-1}}\otimes \sum_{j=1}^{N}a_{i,j_{1},\ldots,j_{i-1},X_{j_{1}},\ldots,X_{j_{i-1}},j}\ket{j}
    \end{align*}
    The algorithm can be expressed in general as \cref{alg:totalQPSD}.

    \begin{algorithm}
	\caption{General pseudo-deterministic quantum algorithm} \label{alg:totalQPSD}
    \raggedright\quad\textbf{Input}: $X\in\{0,1\}^N$.
	\begin{algorithmic}[1]
            \State Prepare the state $\ket{0}$ across the other registers.
		  \For {$k=1\ldots q$}
                \State Apply $V_{i}$ on registers $B_{1}C_{1}\ldots B_{k-1}C_{k-1}B_{k}$.
			    \State Apply $U_{X}$ to registers $B_{k}C_{k}$.
		\EndFor
            \State Perform POVM on registers $B_{1}C_{1}\ldots B_{q}C_{q}$ and output result.
	\end{algorithmic} 
\end{algorithm}
    
    The final state before the POVM has dimension $(2N)^{q}$ as it can be expressed as,
    \begin{equation*}
                \sum_{(j_{1},\ldots, j_{q})\in [N]^{q}}a_{1,j_{1}}\ket{j_{1}}\otimes\ket{X_{j_{1}}}\otimes\ldots a_{q,j_{1},\ldots,j_{q-1},X_{j_{1}},\ldots,X_{j_{q-1}},j_{q}}\ket{j_{q} }\otimes\ket{X_{j_{q}}}.
    \end{equation*} 

    This expression exists in a linear subspace that has at most $(2N)^q$ linearly independent basis states (defined by each $(j_1,\ldots,j_q)\in [N]^q$ and $X_{j_i}$). This does not change if we add ancillae or additional workspace.
    
    We now invoke Holevo's bound. Given that $|\mathrm{Im}(f_{\mathcal{A}})|=m$,
    let $\rvX$ be a random variable which is uniformly distributed over the set $L=\{X_{1},\ldots, X_{m}\}$ 
    whose outputs on applying $f_{\mathcal{A}}$ are different on each element of the set $\forall i,j\in [m]:f_{\mathcal{A}}(X_{i})=f_{\mathcal{A}}(X_{j})\implies i=j$. Let $\rvY$ be the random variable denoting the output of $\mathcal{A}$ for which the random variable $\rvX$ is the input. If the inputs are enumerated $(X_{t})_{t\in [m]}$ and the final state on input $X_{t}$ is $\rho_{t}$, then Holevo's bound (\cref{lem:Holevo} for ensemble $(\rho_{t})_{t\in [m]}$ and distribution uniform over $[m]$) states the following,
    \begin{align*}
        I(\rvX:\rvY)&\leq S\left(\sum_{t=1}^{m}\frac{1}{m}\rho_{t}\right)-\sum_{t=1}^{m}\frac{1}{m}S(\rho_{t})\\
        &\leq S\left(\sum_{t=1}^{m}\frac{1}{m}\rho_{t}\right)\\
        &\leq \log\left(\text{dim}\left(\sum_{t=1}^{m}\frac{1}{m}\rho_{t}\right)\right)\\
        &\leq q\log(2N).
    \end{align*}
    
    The second line follows as each $S(\rho_{t})$ is non-negative. The third follows by the maximum entropy of a state over dimension $d$ (\cref{rem:maxEnt}). 

    The left hand side can be bounded as follows,
    
    \begin{align*}
        I(\rvX:\rvY)&=H(\rvX)-H(\rvX\vert\rvY)\\
        &=\log(m)-H(\rvX\vert\rvY)\\
        &\geq \log(m) - H\left(\frac{1}{3},\frac{2}{3}\right)-\frac{1}{3} \log(m)\\
        &\geq \frac{2}{3}\log(m)-2.
    \end{align*}
    The third line follows from Fano's inequality (\cref{lem:Fano}).

    Therefore, $q\log(2N)\geq I(\rvX:\rvY)\geq \frac{2}{3}\log(m)-2$, which implies $m=O((2N)^{3q/2})$.
    
\end{proof}

Now we are ready to prove \autoref{thm:TotalPoly} with a similar search-to-decision argument as in the classical case (\cite[Theorem 4.1]{GGR13}). Note that the complexity differs from that in the classical case, as classically turning a decision problem from randomized to deterministic costs a cubic overhead in complexity ($\forall f:D(f)=O(R(f)^3)$), whereas in the quantum setting this is quartic ($\forall f:D(f)=O(Q(f)^4)$).

\begin{proof}[Proof of \cref{thm:TotalPoly} ]
    Fix a pseudo-deterministic quantum algorithm $\mathcal{A}$ with query complexity $\psQ(R)$. 
    For every $X$, there exists a canonical output $Y_{X}$ of $\mathcal{A}$. By \cref{lem:MaxOutput}, we can encode this output in $d=O(\psQ(R)\log(N))$ bits.
    We can then define $d$ total boolean functions $(f_{i})_{i}^{d}$, where for all $i\in [d]: f_{i}(X)=(Y_{X})_{i}$, i.e., $f_i$ outputs the $i$-th bit of $Y_{X}$.

    We use the following fact relating the deterministic query complexity with the quantum query complexity of a decision problem.
    
    \begin{theorem}[Theorem 2 \cite{ABKRT21}]
        For any total boolean function $f$, the deterministic and quantum query complexity are related as follows:
        \begin{equation*}
            \D(f_{i})=O(\Q(f_{i})^{4})
        \end{equation*}
    \end{theorem}
    
    We also note that $\Q(f_{i})=O(\psQ(R))$ for all $i \in [d]$, from which we conclude that \[\D(R)=O\left(\sum_{i=1}^{d}\D(f_{i})\right)=O(d\cdot \psQ(R)^{4})=O(\psQ(\probR)^{5}\log(N)).\]
\end{proof}

\subsection{General pseudo-deterministic quantum query upper bound}

We present a general quantum pseudo-deterministic query upper bound for any problem via its randomized query complexity and verification complexity.

\FindoneQComplete*

Our proof is a straightforward generalization of \cite[Theorem 3]{GIPS21}, which shows the \emph{completeness} of \findone (for classical pseudo-deterministic query complexity) by building a general reduction from any problem to \findone. We observe that the same reduction also works in the quantum setting, and then combine it with the $O(\sqrt{N})$ upper bound of \findone~\cite{GIPS21}(see also \cref{thm:PsDGrov}). We present a full proof of \ref{thm:Find1QComplete} in \cref{sec:find1QComplete}.

\section{Query complexity separations}\label{sec:Sep}
In this section, we prove two query complexity separation results for pseudo-deterministic quantum algorithms.
In \cref{sec:AOES}, we present a problem that is trivial for classical randomized algorithms but maximally hard for pseudo-deterministic quantum algorithms. 
In \cref{sec:SimonHamming}, we provide a problem that is easy for both classical randomized algorithms and pseudo-deterministic quantum algorithms, but difficult for classical pseudo-deterministic algorithms. 

\subsection{A maximally hard problem for pseudo-deterministic quantum algorithms}\label{sec:AOES}

We define the Avoid One Encrypted String (AOES) problem and prove that its pseudo-deterministic quantum query complexity is $\Omega(N)$, while being completely trivial for standard randomized algorithms. We prove the separation of \cref{thm:IntroAOES} that we restate below.

\IntroAOES*

We use the $\XOR_N$ problem as our building block, since $\XOR_N$ is (asymptotically) {maximally hard} for quantum query algorithms. Here, for $a,b\in \{0,1\}: a\oplus b$ is the exclusive logical-or of the two bits, and for strings $X,Y\in \{0,1\}^r$, $X\bigoplus Y$ is the bitwise exclusive logical-or of the two strings. We write $\XORfunc{X}$ as the exclusive or of every bit in the string $X$. 

Here, the input of $\XOR_N$ is a string of length $N$ and the output is the $\XOR_N$ of each bit of the string, we write this problem as \XOR without the subscript when the length is clear from context. A quantum algorithm for \XOR cannot gain \emph{any} advantage over a random guess with fewer than $N/2$ queries, as captured by the following lemma. 

\begin{lemma}[\cite{farhi1998limit, BBCMW01}]\label{Lem:XOR_LB}
    Suppose $\alg$ is a quantum algorithm for the XOR problem which has advantage over a random guess. In other words, if there exists $p,\delta\in \mathbb{R}$ such that for any $X\in\{0,1\}^{N}$,
    \begin{equation*}
        \mathbb{P}[\alg^{X}=1 | (\XORfunc{X}=1)]\geq p+\delta,
    \end{equation*}
    and
    \begin{equation*}
        \mathbb{P}[\alg^{X}=1 | (\XORfunc{X}=0)]\leq p.
    \end{equation*}
    Then $\alg$ must have query complexity at least $\frac{N+1}{2}$.
\end{lemma}

We now define our \AOESfull (\AOES) problem. 
Informally, the input of \AOES uses $m$ separate \XOR instances to encrypt a string $b \in \{0,1\}^m$, and the goal is to avoid outputting the encrypted string $b$.  The intuition is that since the XOR problem is maximally hard for quantum algorithms, the string $b$ cannot be decrypted, and hence the quantum algorithm cannot avoid it.
Note that the parameter $m$ controls a trade-off between the quantum pseudo-deterministic query complexity and the classical randomized query complexity. We obtain the maximal separation by choosing $m$ to be a large constant.

\begin{definition}[\AOESfull (\AOES)]\label{def: AOES}
    Let $m \in \mathbb{N}$. In the $\AOES_{m}$ problem, we are given oracle access to $m$ different strings $X_1,\ldots,X_m\in \{0,1\}^{N/m}$. For each $i\in[m]$ define $b_i = \XORfunc{X_i}$. 
    
    The goal is to output a string $x\in \{0,1\}^{m}$ such that $x \neq b$.
\end{definition}

We write \AOES without the subscript when the parameters of the problem are clear from context. We also assume that $m$ divides $N$ and otherwise we can use a standard padding argument. We begin by noting that, for small values of $m$, \AOES is easy for randomized algorithms and hard for deterministic algorithms

\begin{claim}
    The following holds for the complexity of $\AOES_{m}$.
    \begin{itemize}
        \item $R(\AOES_{m})=O(1)$;
        \item $D(\AOES_{m})=\Theta(N/m)$.
    \end{itemize}

\end{claim}
\begin{proof}
    The first item follows by uniformly randomly sampling a string $x\in\{0,1\}^m$ which will be a valid solution with probability $1-1/2^m$. The upper bound for the second item follows by calculating $b_1=\XORfunc{X_1}$ using $N/m$ queries and outputting $(1\oplus b_1)0^{m-1}$. 
    
    We prove the lower bound using a standard adversary argument. Suppose there exists a deterministic algorithm $\alg$ for \AOES that takes $q=N/m-1$ deterministic queries. Whenever $\alg$ queries a bit, we set that bit equal to $0$. After $\alg$ outputs a solution $x \in \{0,1\}^m$, there is at least one bit in each string $X_i$ that is not queried by $\alg$. Therefore, one can set those remaining bits such that $\XORfunc{X_i} = x_i$ for all $i \in [m]$, which implies $\alg$ is wrong.
\end{proof}

We show that any pseudo-deterministic quantum algorithm for \AOES cannot do much better than the trivial deterministic algorithm.

\begin{theorem}[psQ lower bound for AOES]\label{thm: AOES}
Let $m\in [N]$. Any pseudo-deterministic quantum algorithm for $\AOES_m$ requires $\Omega(N/m^2)$ queries.
\end{theorem}
Note that by setting $m$ in \autoref{thm: AOES} to be a sufficiently large constant, we obtain \autoref{thm:IntroAOES}.
We prove \autoref{thm: AOES} in two steps. We first transform a {zero-error} quantum algorithm for \AOES into a quantum algorithm for the \XOR problem with some \emph{non-zero} advantage. We then replace the zero-error algorithm with a pseudo-deterministic one.

The following lemma demonstrates that access to a single sample from an algorithm for \AOES is enough to extract non-zero information about the \XOR of a string.
\begin{lemma}\label{lem: AOES to XOR}
    If there exists a zero-error quantum algorithm $\mathcal{B}$ for $\AOES_{N,m}$ with query complexity $q$, then there exists a quantum algorithm $\mathcal{A}$ for $\XOR_{N/m}$ which calls $\mathcal{B}$ once such that for every input $ X\in \{0,1\}^{N/m}$,
    \begin{equation*}
        \mathbb{P}\left[\mathcal{A}^{X}=\XORfunc{X} \right]\geq \frac{1}{2} +\frac{1}{2^{m+1}}.
    \end{equation*}
\end{lemma}

    \paragraph{Reduction.}
    We use $\mathcal{B}$ to construct an algorithm $\mathcal{A}$ for $\XOR_{N/m}$. In this algorithm, we encode the \XOR instance $X$ into a randomized instance of $\AOES_m$, which is  information theoretically indistinguishable to a uniformly random instance from $\mathcal{B}$'s perspective. If $\mathcal{B}$'s output meets certain conditions, we can recover the parity of $X$ with certainty.
    The pseudocode of algorithm $\alg$ is described in \cref{alg:AOEStoXOR}.
    \begin{algorithm}
	\caption{Algorithm $\mathcal{A}^X$ for $\XOR_{N/m}$ induced from algorithm $\mathcal{B}$ for $\AOES_{N,m}$} \label{alg:AOEStoXOR}
	\begin{algorithmic}[1]
            \State Uniformly sample an index set $I\subseteq [m]$ and strings $Y_{1},\cdots Y_{m}\in \{0,1\}^{N/m}$.
		\State  For all $j\in [m]$: Let $Z_{j}=\begin{cases}
        Y_{j} & \text{if } j\in I\\
        Y_{j}\oplus X & \text{otherwise }
    \end{cases}$
            \State Run algorithm $\mathcal{B}$ on $(Z_{1}, \cdots, Z_{m})$, get the outcome $a=(a_{1}, \cdots, a_{m})$.
            \If{ $\exists t\in \{0,1\}, \forall j\in [m]:a_{j}=\begin{cases}
        \XORfunc{Y_{j}} & \text{if } j\in I\\
        \XORfunc{Y_{j}}\oplus t & \text{otherwise }\end{cases}$ \label{line: condition}
            }
                    \State Output $1\oplus t$. \label{line: output}
                \Else
                    \State Output a uniformly random bit.
                \EndIf
	\end{algorithmic} 
\end{algorithm}

    \paragraph{Analysis.}
    We now show the correctness of algorithm $\alg$.
    For simplicity, we define the following bits 
    \begin{align*}
        x^*&=\XORfunc{X}, \,
       \forall j\in [m]: y_{j}=\XORfunc{Y_j}, \, z_{j}=\XORfunc{Z_j}.
    \end{align*}

    After the first three steps of algorithm $\alg$, we get $(a_{1}, \cdots, a_{m})$ as a valid solution to the \AOES instance $(Z_1, \ldots, Z_m)$. Now, we define the following set \[J \coloneqq \{j \in [m]: a_j =z_j  \}.\]
    Note that $\alg$ cannot calculate $J$ directly, because $\alg$ does not know the value of $z_j$ when $j \not\in I$. Nevertheless, $\alg$ can always test whether $J$ is equal to $I$, as shown in the following claim.

    \begin{claim}\label{claim: condition}
        The condition at \cref{line: condition} is satisfied if and only if $I = J$. 
    \end{claim}

    \begin{proof}
        For one direction, if $I=J$, then for any $j \in [m]$, it holds that \[a_{j}=\begin{cases}
        z_j & \text{if } j\in I\\
        z_j \oplus1 & \text{otherwise }\end{cases}.\]
        
        Now replace $z_j$ with $y_j$ if $j \in I$, or with $(y_j \oplus x^*)$ otherwise, we get 
        \[a_{j}=\begin{cases}
        y_j & \text{if } j\in I\\
        y_j \oplus x^* \oplus1 & \text{otherwise }\end{cases}.\]
        
        Therefore, the condition at \cref{line: output} is satisfied by taking $t = x^* \oplus 1$.

        For the other direction, if the condition at \cref{line: condition} holds, we can rewrite this condition by replacing $y_j$ with $z_j$ or $(z_j \oplus x^*)$ for $j \in I$ or $j \not\in I$ respectively; it holds that
        \[ a_{j}=\begin{cases}
        z_j & \text{if } j\in I\\
        z_j \oplus x^* \oplus t & \text{otherwise }\end{cases}
        .\]
        
        We claim that $x^* \oplus t = 1$; otherwise, we have $a_j = z_j$ for all $j \in [m]$, which contradicts to the fact that $a$ is a valid solution to the \AOES instance. Therefore, we have $I=J$ by the definition of $J$.
    \end{proof}

    The proof of \autoref{claim: condition} immediately implies the following corollary.

    \begin{claim}\label{claim: certainty}
        Algorithm $\alg$ always outputs the correct answer at \cref{line: output}.
    \end{claim}

    \begin{proof}
        In the proof of \autoref{claim: condition}, we show that if the condition at \cref{line: condition} holds, it must be the case that $x^* \oplus t = 1$. Thus, $x^* = t \oplus 1$ is the correct answer.
    \end{proof}

    \begin{claim} \label{claim: prob}
        The condition at \cref{line: condition} is satisfied with probability $1/2^m$.
    \end{claim}

    \begin{proof}
        By \cref{claim: condition}, it suffices to show that $I=J$ happens with probability $1/2^m$.
        Note that $I$ is chosen independently with $Y_1, \ldots, Y_m$, and $Y_1, \ldots, Y_m$ is uniformly random. Therefore, $I$ is also independent with $(Z_1, \ldots, Z_m)$. 
        Since the set $J$ only depends on $(Z_1, \ldots, Z_m)$ and the internal randomness of $\mathcal{B}$, $I$ and $J$ are also independent. Therefore, $I$ is equal to $J$ with probability $1/2^m$.
    \end{proof}

    Combining \autoref{claim: certainty} and \autoref{claim: prob}, the overall success probability of $\alg$ is 
    \[\frac{1}{2^m} + \frac{1}{2}(1 - \frac{1}{2^m}) = \frac{1}{2} +\frac{1}{2^{m+1}},\]
    which concludes \autoref{lem: AOES to XOR}.

Finally, we show that we can replace the zero-error algorithm for \AOES in \autoref{lem: AOES to XOR} with a pseudo-deterministic algorithm with a small overhead on the query complexity.

\begin{lemma}\label{cor:zeroErrorPsD to PsD}
    If there exists a pseudo-deterministic quantum algorithm $\mathcal{B}$ for $\AOES_{N,m}$ with query complexity $q$, then there exists a quantum algorithm $\mathcal{A}$ for $\XOR_{N/m}$ with query complexity $O(q\cdot m)$ such that for every input $ X\in \{0,1\}^{N/m}$,

    \begin{equation*}
        \mathbb{P}\left[\mathcal{A}^{X}=\XORfunc{X}\right]\geq \frac{1}{2} +\frac{1}{2^{m+2}}.
    \end{equation*}
\end{lemma}

\begin{proof}
    Given such a pseudo-deterministic algorithm $\mathcal{B}$, we can construct a new algorithm $\mathcal{B}'$ that repeatedly runs $\mathcal{B}^X$ for $100m$ times and outputs the most common outcomes from these repetitions (breaking ties arbitrarily). 
    By the Chernoff bound, the error probability of $\mathcal{B}'$ is smaller than $\frac{1}{2^{m+2}}$. 

    Since the zero-error algorithm for \AOES is only called once in \cref{lem: AOES to XOR}, we can replace it by $\mathcal{B'}$, and only increase the overall error probability by at most $\frac{1}{2^{m+2}}$.
\end{proof}

Combining \autoref{cor:zeroErrorPsD to PsD} with \cref{Lem:XOR_LB}, we conclude \autoref{thm: AOES}.

\begin{remark}
    The only place we use the pseudo-deterministic property for proving \autoref{thm: AOES} is the repetition trick in \autoref{cor:zeroErrorPsD to PsD}. In contrast, the trivial random sampling algorithm for \AOES has error probability $1/2^m$; however, we cannot reduce this error via repetition. 
    
    Indeed, our proof of \autoref{thm: AOES} implies that achieving an error probability that is much lower than $1/{2^m}$ for \AOES is extremely difficult, even for a standard quantum algorithm.
\end{remark}

\subsection{Exponential pseudo-deterministic quantum advantage}\label{sec:SimonHamming}

In this section, we define a problem, {Quantum-Locked Estimation} (\SimonHamming), which is hard for classical algorithms to solve pseudo-deterministically, but is easy for pseudo-deterministic quantum algorithms and randomized classical algorithms. 

\IntroSH*

The \SimonHamming problem can be seen as
 a \emph{lock-and-key} construction. We first pick the \Hamming problem --- calculating an estimate of the Hamming weight of a string --- as the `lock', since it is easy for classical randomized algorithms but hard for any pseudo-deterministic classical algorithm. We then add in a `key', the Simon's problem~\cite{Simon97}, as a search problem with exponential quantum advantage, and the unique solution of the `key' problem is a valid solution of the `lock' problem. 
 Combining these, the input to the lock-and-key promise problem is a pair $(f,X)$. An input of this form satisfies the promise (and is therefore a valid input) if the solution of the key problem on input $f$ is a solution to the lock problem on input $X$ and any valid solution overall is a solution to the lock problem on input $X$. This implies that a quantum algorithm can obtain the secret string that is the pseudo-deterministic solution to this problem that we know is classically pseudo-deterministically hard. 

In this section, unless otherwise stated, $N=2^n$. 
We start by introducing the following necessary definitions.

\begin{definition}[Simon's Problem]
    The input to this problem is a function $f:\{0,1\}^n\rightarrow \{0,1\}^n$ promised to be a 2-to-1 or 1-to-1 function with a secret string $s\in \{0,1\}^n$ such that 
    \begin{equation*}
        \forall x,y\in\{0,1\}^n:f(x)=f(y)\iff x=y\oplus s \; \text{ or } \; x=y.
    \end{equation*}
    Given oracle access to this function, the goal is to output the secret string $s$. The special case $s=0^n$ implies that $f$ is a permutation of the strings of length $n$. 
\end{definition}

With some abuse of notation, we interpret $s\in\{0,1\}^n$ to be either a string or an $n$-bit binary number interchangeably. For $s\in \{0,1\}^n$, we denote by $\setS_{s}$ the set of Simon's functions on $\{0,1\}^n$ that encode the secret string $s$, that is all functions $f:\{0,1\}^n\rightarrow \{0,1\}^n$ such that for every distinct $x,y\in\{0,1\}^n$ it holds that $f(x)=f(y)$ if and only if $x=y$ or $ x=y\oplus s$. We shall use the following standard bounds of the query complexity of Simon's problem.
\begin{theorem}[\cite{Simon97}]
    Any classical randomized algorithm for Simon's problem requires $\Omega(2^{n/2})$ queries, whereas there exists a quantum algorithm for Simon's problem using $O(n)$ queries.
\end{theorem}

Indeed, we require the following stronger statement regarding the hardness of Simon's problem.

\begin{lemma}\label{lem: hardness of Simons}
    Given a function $f \in \setS_{s}$
    where the secret string $s$ is promised to be in one of two disjoint sets $S_1, S_2\subseteq \{0,1\}^n$, then any classical randomized algorithm needs $\Omega(\sqrt{\min(|S_1|,|S_2|)})$ queries to determine whether $s$ is in $S_1$ or $S_2$.
\end{lemma}

\begin{proof}
    With half the probability, let the secret string $s$ be uniformly drawn from $S_1$; and with half the probability, let $s$ be uniformly drawn from $S_2$. Then, take $f$ to be uniformly drawn from $\setS_s$.
    
    By Yao's minimax principle, it suffices to consider any deterministic algorithm $\alg$ making at most $t = o(\sqrt{\min(|S_1|,|S_2|)})$ queries on the input distribution above. Without loss of generality, we assume $\alg$ will stop and output the answer immediately once it finds a collision of $f$.
    By symmetry, we can simplify the decision tree of $\alg$ by merging nodes in a way that each non-leaf node has only two children: one represents that a collision is found, and the other represents that no collision is found.

    We call a string $t$ at node $u$ \emph{alive} if the queries to $u$'s ancestors have not ruled out the possibility that the secret string of $f$ being $t$. All the strings in $S_1$ and $S_2$ are alive at the root, and a query at depth $d$ can \emph{kill} at most $d-1$ strings. By symmetry, the probability that $\alg$ never finds a collision is $1 - o(1)$. And conditioned on reaching a leaf $u$ without collision, most of the strings in both $S_1$ and $S_2$ are still alive at $u$. Thus, no matter $\alg$ outputs $S_1$ or $S_2$ at $u$, $\alg$ will be wrong with probability at least $\frac{1}{2} - o(1)$. This concludes that $\alg$ is not correct, and thus $\Omega(\sqrt{\min(|S_1|,|S_2|)})$ queries are necessary.
\end{proof}

We now formally state the version of the Hamming weight approximate problem that we will denote by \Hamming.

\begin{definition}[The Hamming Problem]
    Given a bit string $X\in\{0,1\}^N$, the goal is to output an approximation of the Hamming weight of $X$ within an additive error range of $N/10$; that is,
    \begin{align*}
        \Hamming=\{(X,t) \in\{0,1\}^N \times \mathbb{N}: \big| \|X\|_{\mathsf{hw}}-t \big| \leq N/10\} \;.
    \end{align*}
\end{definition}

Since we allow a solution to \Hamming to be within a range of the actual Hamming weight, it can be easily solved by a random sampling algorithm in $O(1)$ queries.
However, \Hamming was proven to be maximally hard for classical pseudo-deterministic algorithm in \cite{GGR13}. We reproduce the proof here, as it includes some useful ideas for proving the lower bound of the \SimonHamming problem later.

\begin{theorem}[Corollary 3.2 of \cite{GGR13}]\label{thm: hardness of hamming}
    $\psR(\Hamming) = \Omega(N)$.
\end{theorem}

\begin{proof}
    
    Now let $\alg$ be a pseudo-deterministic algorithm for \Hamming.
    Let $X^*\in\{0,1\}^N$ be some fixed string whose Hamming weight is $N/2$, and take
     $t \coloneqq \alg^{X^*}$. We then pick a string $Y\in\{0,1\}^N$ whose Hamming weight is the smallest among those for which $\alg$ outputs $t$ with probability at least $2/3$ (break ties arbitrarily).
    The Hamming weight of $Y$ is at least $3N/10$.
    
    Now consider the set $\Flip_Y$ of strings by flipping a ``1'' to ``0'' in $Y$, i.e., 
    \[\Flip_Y \coloneqq \{X \in \{0,1\}^N: \exists i \in [N], X_i = 0 \land Y_i = 1 \land (X_j = Y_j,  \forall j \neq i) \}.  \]
    By definition, we have $\alg^Y \neq \alg^{Y'}$ for any string $Y' \in \Flip_Y$; we also know $|\Flip_Y| \geq 3N/10$.
    Therefore, $\alg$ requires $\Omega(N)$ queries to distinguish $Y$ from all the strings in $\Flip_Y$. 
\end{proof}

We now state the \SimonHamming problem, where the Simon's problem is the `key' and \Hamming is the `lock'. 

\begin{definition}[The \SimonHamming Problem]
    The input is a pair $(f,X)$ for $f:\{0,1\}^n\rightarrow \{0,1\}^n$ and $X\in \{0,1\}^N$ where $N=2^n$. Here, $(f,X)$ is promised to satisfy the following.
    \begin{itemize}
        \item $f \in \setS_{s}$ for a string $\bv{s}\in \{0,1\}^n$, i.e., $f$ is a valid Simon's instance encoding the secret $s$;
        \item $\big| \|X\|_{\mathsf{hw}} - \bv{s}\big| \leq 0.09\cdot N$, i.e., the value of $\bv{s}$ is a valid solution to the \Hamming instance $X$.
    \end{itemize}
    
    The goal is to output an integer $t\in[N]$ within $N/10$ distance of the Hamming weight of $X$. 
\end{definition}

In this case, a query to an input $(f,X)$ of \SimonHamming is either a query to $f$ or to $X$ and the query complexity of an algorithm for \SimonHamming is the total number of queries to either part of the input. We show that \SimonHamming is easy for classical randomized and a quantum pseudo-deterministic algorithm exists that canonizes the classical one, but the problem is hard for classical pseudo-deterministic algorithm; this immediately implies \autoref{thm:IntroSH}. 

\begin{theorem}\label{thm: SHComp}
    Let $N = 2^n$ be the length parameter in \SimonHamming, then the following holds for the complexity of \SimonHamming.
    \begin{enumerate}
        \item $\R(\SimonHamming) = O(1)$;
        \item $\psQ(\SimonHamming) = O(\log N)$;
        \item $\psR(\SimonHamming) = \Omega(\sqrt{N})$. \label{item: psR}
    \end{enumerate} 
    Furthermore, there exists a classical randomized algorithm $\mathcal{A}$ for \SimonHamming with query complexity $O(\log N)$ and a pseudo-deterministic quantum algorithm $\mathcal{B}$ with query complexity $O(\log(N))$  such that $\mathcal{B}$ canonizes $\mathcal{A}$.
\end{theorem}

\paragraph{Upper bounds and canonization.}
    The first upper bound follows as a classical randomized algorithm can solve the \Hamming part directly via sampling and ignore the Simon's problem with $O(1)$ queries via random sampling. For the second upper bound, note that a pseudo-deterministic quantum algorithm can also solve \SimonHamming using $O(\log N)$ queries by first solving for the secret string $s$ given $f$ in the Simon's problem, and then outputting $\bv{s}$ as the estimate of the Hamming weight.

    For the last remark on canonization (\cref{def: Elevate}), we take the quantum algorithm $\mathcal{B}$ to be one that solves and outputs the secret string $s$ as described above.
    However, we need to construct a modified randomized algorithm $\alg$, described in below:
    \begin{enumerate}
        \item Make $10^6\cdot \log N$ queries to $X$ and let $\hat{t}$ be its best estimation of the Hamming weight of $X$.
        \item Output a uniformly random integer in the range $I\coloneqq[\hat{t}-0.095N,\hat{t}+0.095N]$.
    \end{enumerate}

    By the Chernoff bound, the probability that $\big|\hat{t} - \|X\|_{\mathsf{hw}} \big| \leq 0.005N$ is $\big(1 - o(1/N)\big)$. In that case, $\alg$ always outputs a correct solution, and we have $s \in I$ by the promise of the \SimonHamming problem.
    Thus, the canonical output $s$ of $\mathcal{B}$ will be outputted by $\alg$ with probability $\Omega(1/N)$. This concludes that $\mathcal{B}$ canonizes $\alg$

\paragraph{Pseudo-deterministic lower bound.}
    To show the $\Omega(\sqrt{N})$ lower bound for any classical pseudo-deterministic algorithm, we claim that any sub-linear pseudo-deterministic classical algorithm for \SimonHamming must \emph{approximately} solve the underlying Simon's problem, implying a lower bound from the classical hardness of Simon's.

\begin{lemma}\label{lem: reduce Simon to SH}
    Let $\alg$ be a pseudo-deterministic classical algorithm for \SimonHamming that uses $o(N)$ queries. For any string $\Xh\in \{0,1\}^N$ such that $\|\Xh\|_{\mathsf{hw}}=N/2$ and any string $s\in \{0,1\}^n$ whose value is between $0.41N$ and $0.59N$, it holds that for any Simon's function $f\in\setS_s$, the canonical output of $\alg$ on input $(f,\Xh)$  is some $t$ for which $|s-t|\leq0.01N$. 
\end{lemma}

\begin{proof}
    Take $\alg, \Xh, s$ as objects that satisfy the conditions in the statement of \autoref{lem: reduce Simon to SH} and consider an arbitrary $f \in \setS_s$ and let $t$ be the canonical output of $\alg$ on input $(f, \Xh)$. By the correctness of $\alg$, we have $2N/5 \leq t \leq 3N/5$. 
    
    We then pick a string $\Ylow$ whose Hamming weight is the smallest given that $\alg^{(f, \cdot)}$ outputs $t$ (break ties arbitrarily). The Hamming weight of $\Ylow$ is at least $2N/5-N/10=3N/10$. We claim that $\|\Ylow\|_{\mathsf{hw}} = s - 0.09N$.
    
    Suppose otherwise, i.e., $\|\Ylow\|_{\mathsf{hw}} > s - 0.09N$, we define the set $\Fliplow$ of strings reached by flipping any ``1'' to a ``0'' in $\Ylow$. For any string $Y' \in \Fliplow$, we have $\|Y'\|_{\mathsf{hw}} \geq s - 0.09N$, so $(f, Y')$ is a valid input to the \SimonHamming problem. 
    Recall that $t$ is the canonical output of $\alg$ on input $(f, \Ylow)$ and let $t'$ be the canonical output of $\alg$ on input $(f, Y')$. By definition of $Y_{\mathsf{low}}$, we have that $t \neq t'$. We know $|\Fliplow| \geq 3N/10$, which means $\alg$ needs at least $\Omega(N)$ queries to distinguish $\Ylow$ with all the strings in $\Fliplow$. This contradicts the assumption that $\alg$ uses $o(N)$ queries. Thus, we conclude that $\|\Ylow\|_{\mathsf{hw}} = s - 0.09N$.

    Similarly, we can take $\Yhigh$ to be the string with largest Hamming weight such that $\alg^{(f, \cdot)}$ outputs $t$. By symmetry, we have $\|\Yhigh\|_{\mathsf{hw}} = s + 0.09N$. Since $t$ is a valid solution for both $\Yhigh$ and $\Ylow$, we have that $|t-(s-0.09N)|\leq N/10$ and $|t-(s+0.09N)|\leq N/10$. This implies $|s-t|\leq 0.01N$.
\end{proof}

    Now we are ready to prove the lower bound for classical pseudo-deterministic algorithms.

\begin{proof}[Proof of \autoref{thm: SHComp}, Item \ref{item: psR}]
    Define the sets \[S_1 = \{s: 0.4N\leq s \leq 0.45N\},\; S_2 = \{s: 0.55N\leq s \leq 0.6N\}.\] We know $|S_1|=|S_2| = N/20=\Omega(N)$.
    
    Suppose the secret string $s$ is promised to be in one of $S_1, S_2$ and the goal is to distinguish between these two cases, \cref{lem: hardness of Simons} shows that the randomized query complexity of this task is still $\Omega(\sqrt{N})$.
    
    Now, let $\alg$ be any classical pseudo-deterministic algorithm for \SimonHamming with $o(N)$ query complexity. If given a Simon's function $f \in \mathcal{F}_s$ such that $s$ is in $S_1$ or $S_2$, we can run $\alg^{(f,\Xh)}$ for some $\Xh$ with $\|\Xh\|_{\mathsf{hw}}=N/2$, and take $t$ to be its pseudo-deterministic solution. By \autoref{lem: reduce Simon to SH}, we must have $t \leq 0.46N$ or $t \geq 0.54N$, which allow us to determine whether $s \in S_1$ or $s \in S_2$. Therefore, the query complexity of $\alg$ is $\Omega(\sqrt{N})$.

\end{proof}

\section{Pseudo-deterministic quantum advantage with small overhead}\label{sec:GrovWit}
In this section, we show a large class of quantum query algorithm can be made pseudo-deterministic with small overhead with a \emph{binary search} based approach. We first present a pseudo-deterministic Grover algorithm with no overhead asymptotically in \autoref{sec:GrovWit:Gro}. 
We then generalize the idea to a broader class of algorithms in \autoref{sec:GrovWit:Wit}.

\paragraph{Notation.} For any index set $I \subseteq [N]$, let $X\vert_{I} \in \{0,1\}^{|I|}$ be the subsequence of $X$ restricted to indices in $I$.
For any two tuples of elements $X=(x_{1},\ldots,x_{t}), Y=(y_{1},\ldots,y_{\ell})$, we write $\concat{X}{Y}=(x_{1},\ldots,x_{t},y_1, \ldots, y_{\ell})$ as their concatenation.

\subsection{Warm-up: Grover's algorithm}\label{sec:GrovWit:Gro}

We start with an informal result of \cite{GIPS21} that Grover's algorithm can be made pseudo-deterministic via a \emph{binary search} at the cost of logarithmic overhead in query complexity. We employ a more careful analysis to show that this overhead can be avoided. 
Their result, as stated in \cref{thm:find1_lb}, provides a lower bound of $\Omega(\sqrt{N})$ even when the majority of the input's bits are $1$.

Let \Grover denote the quantum query algorithm which on input $X\in\{0,1\}^N\setminus \{0^N\}$ outputs $i\in [N]$ such that $X_i=1$ that uses $O(\sqrt{N})$ queries.
We formally present the pseudo-deterministic Grover algorithm in \cref{alg:PsD_Grover}.

\begin{algorithm}
	\caption{Pseudo-deterministic Grover search for $N=2^n$ elements} 
    \label{alg:PsD_Grover}
	\begin{algorithmic}[1]
            \State Let $\vL \leftarrow 1, \vR \leftarrow N$
		  \For {$k=0\ldots n-1$}
                \State $\vcur \leftarrow \lfloor (\vL + \vR)/2 \rfloor, I \leftarrow \{\vL, \vL+1, \ldots, \vcur\}$
			    \State Run $\Grover(X\vert_{I})$ for $k+2$ times.
                \If{ at least one call to \Grover returns an index $i \in I$ such that $X_i = 1$:}
                    \State $\vR \leftarrow \vcur$
                \Else
                    \State $\vL \leftarrow \vcur+1$
                \EndIf
		\EndFor
            \State Output $\vL$
	\end{algorithmic} 
\end{algorithm}

\begin{theorem}\label{thm:PsDGrov}
    For $N >1$, There is a pseudo-deterministic quantum algorithm that has query complexity $O(\sqrt{N})$, and on input $X\in \{0,1\}^{N}\backslash \{0^{N}\}$, it outputs the smallest index $i\in [N]$ such that $X_{i}=1$ with probability at least $2/3$.
\end{theorem}

\begin{proof}
In this proof, without loss of generality, we assume $N=2^n$ for some $n\in \mathbb{N}$ as otherwise we can pad the input at a cost of a constant number of queries. We consider the algorithm presented in \cref{alg:PsD_Grover}. First note that the query complexity of \cref{alg:PsD_Grover} can be upper bounded by
    \begin{equation*}
        \sum_{k=0}^{n-1}(k+2)O\left(\sqrt{N/2^{k+1}}\right)=O\left(\sqrt{N}\sum_{k=0}^{n-1}\frac{k+2}{\sqrt{2^{k+1}}}\right)=O(\sqrt{N}).
    \end{equation*}

    Let $i_{0}$ be the earliest index such that $X_{i_{0}}=1$ (this exists by the promise on the input). The value of $L$ at the end of the algorithm will be equal to $i_0$ if and only if $i_{0}\in [L,R-1]$ at every stage of the algorithm. If we assume that one of the iterations of \Grover accepts when $i_0\in I$, then we have that every reassignment of the set $I$ halves the size of $I$ and will still contain $i_0$. Therefore, after $n$ iterations when the algorithm terminates, $\vL=i_0$, $\vR=i_0+1$, and $I=\{i_0\}$.

    The binary search in \cref{alg:PsD_Grover} may fail only if there exists a step $k\in[n-1]$ such that all of the $k+2$ repetitions of Grover's algorithm fail to find an existing ``1'' in $X\vert_{I}$ when $i_0\in I$. At each stage, if it does find a $1$, then $i_{0}$ remains in $I$ at the next step. The probability all $k+2$ iterations of Grover fail at step $k$ is upper bounded by $1/3^{k+2}$. By a union bound over all $k$, we have the error probability upper bounded by $\sum_{k=0}^{n-1}\frac{1}{3^{k+2}} < 1/6$.
\end{proof}

\subsection{\Prunable problems}\label{sec:GrovWit:Wit}

We now generalize the idea in the previous section and present a broader class of problems such that any quantum algorithm for a problem in this class can be made pseudo-deterministic with small overhead. In particular, we consider problems that look for a small subset of the input (i.e. a witness) that has a certain property.

\begin{definition}[$k$-subset finding \cite{CE05subset}]
    Let $k,N\in\mathbb{N}$. The $k$-subset finding problem is defined by a relation $\probR \subseteq \Sigma^{k}\times [N]^{k}$. Given oracle access to an input $X\in \Sigma^{N}$, the task is to find a pair $(Y,I)$ such that $I\subseteq[N]$ is of size $k$ and $X\vert_{I}=Y$ such that $(Y,I) \in \probR$; or output $\bot$ if no such pair exists.
\end{definition}

We now describe a mild condition where we describe a class of problems as \Prunable, which allows an algorithm $\alg$ for a $k$-subset finding problem $\probR$ to be made pseudo-deterministic with small overhead. Informally speaking, a problem $\probR$ is \Prunable if for every $I\subseteq [N]$ there exists an algorithm that finds a subset for $\probR$ when restricted to only having indices in $I$ with query complexity at most that of finding a subset for $\probR$.

\begin{definition}[\Prunable]
    Let the search problem defined by $\probR \subseteq \Sigma^{k}\times [N]^{k}$ be an instance of the $k$-subset finding problem and let $\alg$ be the $\Q(R)$-query quantum algorithm that solves $\probR$. For any index set $I\subseteq [N]$, define a new $k$-subset finding problem $\probR\vert_I \coloneqq \probR \cap (\Sigma^{k}\times I^{k})$.

    We say that $\probR$ is \Prunable if for any index set $I\subseteq [N]$, there exists an algorithm $\alg_I$ that solves $\probR\vert_I$ with $O(\Q(R))$ queries.
\end{definition}

\prunable*

\begin{proof}

Consider the algorithm presented in \cref{alg:PsD_Prune}. This algorithm proceeds by iteratively building up the lexicographically \emph{last} subset that solves the search problem. Suppose the first $\ell -1 $ elements of the subsets have been found, and let $T$ be the set of indices of these $\ell -1 $ elements. The algorithm performs a binary search to find the largest index $i$ such that the problem $R\vert_{I}$ still has a valid solution, where $I = T\cup\{i,i+1 \ldots, N\}$. Then, $X_i$ is chosen to be the $\ell$-th element of the subset.

\begin{algorithm}[ht]
	\caption{Pseudo-deterministic algorithm for a \Prunable problem $R$ with a randomized quantum algorithm $\alg$} \label{alg:PsD_Prune}
 
	\begin{algorithmic}[1]
            \State Let $\mu \gets 10\log(k\log(N))$
            \State Let $Y \gets (), T \gets \emptyset, \LB \gets 1$
		\For {$\ell=1,\ldots, k$}
  
                \State Let $L\leftarrow \LB, R\leftarrow N$
                \While{$L < R$}
                    \State $\vcur \leftarrow \lceil (L+R)/2\rceil$, $I \gets T\cup\{\vcur, \vcur+1,\ldots, N\}$
			     \State Run $\alg_{I}{(X)}$ for $\mu$ times.
                    \If{ any iteration of $\alg_I$ finds a solution of $R\vert_I$}
                        \State $L\leftarrow \vcur$
                    \Else
                         \State $R\leftarrow \vcur - 1$
                    \EndIf
                \EndWhile
                \State $Y\leftarrow Y \circ X_{L}$
                 \State $T\leftarrow \{L\}\cup T$
                \State $\LB\leftarrow L+1$
		\EndFor
            \State Output $(Y, T)$ 
	\end{algorithmic} 
\end{algorithm}

\autoref{alg:PsD_Prune} may fail only if there exists a step $\ell \in [k]$ such that in one of the binary search steps, $R\vert_I$ has a solution, but
all of the $\mu = 10\log(k\log(N))$ repetitions of $\alg_I$ fail to find one. Note that for each $\ell \in [k]$, the binary search takes $O(\log(N))$ iterations of $\alg_I$.
Thus, by a union bound over $k \cdot \log(N)$ iterations, the failure probability is at most

\begin{equation*}
    \frac{k\log(N)}{3^{10\log(k\log(N))}}= \frac{k\log(N)}{(k\log(N))^{10\log(3)}}=o(1).
\end{equation*}

Similarly, by the \Prunable condition, each repetition of $\alg_I$ takes $O(\QC)$ queries. Thus,
the query complexity of \autoref{alg:PsD_Prune} can be upper bounded by $k \cdot \log N \cdot \mu \cdot \QC = \tilde{O}(k \cdot \QC)$.

\end{proof}

From this result we immediately have the following pseudo-deterministic quantum query algorithm for the following problems which are \Prunable .

\begin{definition}
    $k$-\distinctness is the problem on input $X\in [m]^N$ of finding $i_{1},\ldots,i_{k}$ such that $X$ takes the same value on each index. 
\end{definition}

There exists an $\tilde{O}(N^{3/4})$ query quantum algorithm for $k$-\distinctness \cite[Theorem 8]{Bel12}.

\begin{corollary}
    There is an $\tilde{O}(N^{3/4}k)$ query pseudo-deterministic quantum algorithm for $k$-\distinctness.
\end{corollary}

\begin{definition}
    \Triangle is the problem on a graph input over $n$ vertices of finding a complete subgraph over $3$ vertices.
\end{definition}

There exists an $\tilde{O}(n^{5/4})$ query quantum algorithm for \Triangle in the adjacency matrix model \cite[Theorem 1.1]{Gall14}.

\begin{corollary}
    There is an $\tilde{O}(n^{5/4})$ query pseudo-deterministic quantum algorithm for \Triangle.
\end{corollary}

\begin{definition}
    $k$-\ksum is the problem on query access to an input of the form $X\in [m]^N$ of determining if there is some subset of size $k$ summing to some fixed integer $t$.
\end{definition}

There exists a general $O(N^{k/(k+1)})$ query quantum algorithm for $k$-subset finding in Theorem 1 of \cite{CE05subset}. $k$-\ksum is a special case implying the following.

\begin{corollary}
    There is an $\tilde{O}(N^{k/(k+1)}\cdot k)$ query pseudo-deterministic quantum algorithm for the $k$-\ksum problem.
\end{corollary}

\begin{definition}
    \GraphCollision is the problem, given access to an undirected graph $G$ over $n$ vertices and query access to a set of boolean values $\{x_v\vert v\in G\}$, of determining whether there are two vertices $v,v'$ connected by an edge for which $x_v=x_{v'}=1$.
\end{definition}

There exists an $O(t^{1/6}\cdot\sqrt{n} )$ query quantum algorithm for the \GraphCollision over graphs of tree-width $t$ by Theorem 1 of \cite{GraphColl}.

\begin{corollary}
    There is an $\tilde{O}(t^{1/6}\cdot\sqrt{n})$ query pseudo-deterministic quantum algorithm for the \GraphCollision problem when the tree-width is $t$.
\end{corollary}

\begin{remark}
    Note that the last result follows as the treewidth of any subgraph is at most that of the original graph making this problem \Prunable.
\end{remark}

\paragraph{Limitation on problems with global structure.}

An example of a problem that cannot be made pseudo-deterministic via this theorem is the \findone problem. In particular, the promise that the majority of input bits are one does not hold anymore during the binary search. Informally speaking, if an efficient randomized / quantum query algorithm $\alg$ heavily exploits the \emph{global structure} of the problem $\probR$, then our binary search approach may fail, and likely, $\probR$ does not have any efficient pseudo-deterministic algorithm whatsoever.

We now present another example of a quantum algorithm that cannot be made pseudo-deterministic due to the ``global structure''.

\begin{definition}
    Fix a graph $G=(V,E)$ and give the oracle access to a function $f:V\rightarrow \{0,1\}$ such that a vertex $v \in V$ is marked if $f(v)=1$. The goal is to output a marked vertex.
    
    The hitting time of $f$ is defined to be the expected length of a random walk that terminates on hitting a vertex marked by $f$.
\end{definition}

The quantum walk algorithm achieves a quadratic speed-up compared to a classical random walk.

\begin{theorem}[\cite{Szeg04, AGJK20walk}]
    For quantum algorithms, finding a marked vertex on a graph with hitting time $HT$ requires $\tilde{O}(\sqrt{HT})$ queries.
\end{theorem}

This, however, fails to hold pseudo-deterministically since the promise that is true on the whole graph (small hitting time) does not follow when you take a sub-graph. 
Consider the graph $G$ being the complete graph on $n$ vertices, and a majority of the set of vertices are marked. This is reduced to the $\sqrt{N}$ lower bound for \findone~\cite{GIPS21}.

We note that although the quantum walk algorithm for finding a marked vertex cannot be made pseudo-deterministic, algorithms that use this as a subroutine such as $k$-distinctness \cite{Bel12} can be made pseudo-deterministic as the problem is \Prunable.

\section{Quantum search problems and continuous domains}\label{sec:quantInputs}

We extend the notion of a search problem to quantum inputs. To this end, consider a search problem $\probR\subseteq\mathcal{X}\times \mathcal{Y}$. Here, $\mathcal{X}$ is a collection of pure states,
\begin{equation*}
    \mathcal{X}\subseteq \bigcup_{d\in \mathbb{N}}\mathcal{S}(\mathbb{C}^d),
\end{equation*} 
and $\mathcal{S}(\mathbb{C}^{d})$ denotes the set of pure states of dimension $d\in \mathbb{N}$ expressed as a normalised vector $\ket{\psi}\in \mathbb{C}^d$ such that $\braket{\psi|\psi}=1$.
In this setting, we consider algorithms that receive the dimension $d\in \mathbb{N}$ explicitly and oracle access to copies of the unknown quantum state $\ket{\psi}\in \mathcal{S}(\mathbb{C}^d)$. In this model, a quantum algorithm is a POVM acting on $\ket{\psi}\bra{\psi}^{\otimes t}$ and the value of $t$ is the query complexity of the algorithm (the number of copies of the unknown state $\ket{\psi}$). 

Formally, a POVM is expressed as a set of positive semi-definite matrices $\{M_{y}\}_{y\in \mathcal{Y}}$ that act on $(\mathbb{C}^{d})^{\otimes t}$ and $\sum_{y\in \mathcal{Y}}M_y=I$ where $I$ is the identity matrix on $(\mathbb{C}^{d})^{\otimes t}$. The probability of seeing an output $y\in\mathcal{Y}$ on input $\ket{\psi}\in \mathcal{S}(\mathbb{C}^d)$ is $\mathrm{Tr}(M_y\ket{\psi}\bra{\psi}^{\otimes t})$.

Here, we consider pseudo-determinism to be the same notion as in \cref{def:PsD_alg}.

\begin{remark}
    In this section, we focus on search problems mapping quantum pure states to classical strings; however, one can also consider inputs in the form of mixed states, unitaries, or channels.  

    When the output of a quantum algorithm is itself a quantum state, the notion of pseudo-determinism becomes problematic. In most such tasks, approximate outputs are allowed (e.g. tomography), but repeated runs of the algorithm will yield different approximations for a fixed input. This suggests that there is no immediately apparent definition for pseudo-determinism for these problems that maintains its essential properties.
\end{remark}

Here, we consider search problems whose inputs are quantum states and prove that the number of pseudo-deterministic solutions is upper bounded by the number of connected components of the set of valid inputs. For this, we require the following definition.

\begin{definition}
    Let $d\in\mathbb{N}$. Consider a quantum query algorithm $\alg$ whose input is a quantum state $\ket{\phi}\in\mathcal{S}(\mathbb{C}^d)$ and whose output space is $\mathcal{Y}$. 
    
    We define $g_{\alg}:\mathcal{S}(\mathbb{C}^d)\rightarrow \Delta_{\mathcal{Y}}$ to be the map from pure states to their output distributions under this algorithm. Here, $\Delta_{\mathcal{Y}}$ is the probability simplex over $\mathcal{Y}$. Equivalently,
    \begin{equation*}
        \forall \ket{\phi}\in\mathcal{S}(\mathbb{C}^{d}),\forall y\in \mathcal{Y}: g_{\alg}\left(\ket{\phi}\right)_{y} \coloneqq \mathbb{P}\left[\alg^{\ket{\phi}}=y\right].
    \end{equation*}
\end{definition}

\begin{lemma}\label{lem:AlgContMap}
    Let $d\in\mathbb{N}$. For any algorithm $\alg$ whose inputs are $d$-dimensional pure states, $g_{\alg}$ is a continuous map.
\end{lemma}

\begin{proof}
    Let $\mathcal{Y}$ be the output space of algorithm $\alg$, where $s = |\mathcal{Y}|$, and let $\alg$ have query complexity $t>0$. $\alg$ is described by an $s$-outcome POVM whose input is $t$ copies of a quantum state $\ket{\phi}\bra{\phi}$. The POVM is expressed as $s$ positive semi-definite matrices $\{M_1,\ldots, M_s\}$ for which $\sum_{j=1}^{s}M_j=I$. The output of $g_{\alg}$ is as follows:
\begin{align*}
    \forall y\in \mathcal{Y}:g_{\alg}\left(\ket{\phi}\right)_{y}\coloneqq\mathbb{P}[\alg^{\ket{\phi}}=y]&=\mathrm{Tr}\left(M_y\ket{\phi}\bra{\phi}^{\otimes t}\right).
\end{align*}
The map $\ket{\phi}\mapsto\ket{\phi}\bra{\phi}^{\otimes t}$ is polynomial in the components of $\ket{\phi}$, and hence continuous; Multiplication by a fixed matrix and taking the trace are continuous as well. This implies that for each $ y\in \mathcal{Y}$, the map $\ket{\phi}\mapsto g_{\alg}(\ket{\phi})_y$ must be continuous. Hence $g_{\alg}$ is continuous.
\end{proof}

\begin{corollary}
    Let $s,d\in\mathbb{N}$. Suppose $\alg$ is a pseudo-deterministic algorithm whose input is an element of a set $S\subseteq \mathcal{S}(\mathbb{C}^d)$, and whose output is an element of $[s]$. Let $f_{\alg}:S\rightarrow [s]$ be the function that maps an input state to its pseudo-deterministic solution according to $\alg$ (which exists and is unique by the definition of pseudo-determinism). Then the number of possible outputs, $|\mathrm{Im}(f_{\alg})|$, is at most the number of connected components of $S$.
    
\end{corollary}

\begin{proof}

    It suffices to prove that for any connected component of $S$, the number of pseudo-deterministic solutions when restricted to inputs in this component is $1$. Therefore, without loss of generality, we assume $S$ is connected.

    Since the probability of any given output of a pseudo-deterministic function is either less than $1/3$ or greater than $2/3$, it follows that $\forall j\in[s]:g_{\alg}(\cdot)_j$ maps a connected space $S$ to a subset of $[0,1/3]\cup[2/3,1]$. By \cref{lem:AlgContMap}, this function is continuous and therefore $\mathrm{Im}(g_{\alg}(\cdot)_j)$ is a connected space. 
    
    Since $[0,1/3]$ and $[2/3,1]$ are the only connected subsets of $[0,1/3]\cup[2/3,1]$, for each $j\in[s]$, either $g_{\alg}(\cdot)_j$ maps all of $S$ into $[0,1/3]$ (in which case $j$ is never a pseudo-deterministic solution) or maps all of $S$ into $[2/3,1]$ (in which case $j$ is always the pseudo-deterministic output). Since pseudo-determinism requires $\alg$ to output a unique solution with probability at least $2/3$, exactly one value of $j\in[s]$ is the pseudo-deterministic solution over the entire set of inputs, $S$.

\end{proof}

This implies that any pseudo-deterministic algorithm whose domain is the entire set of quantum states of a fixed dimension can only output a constant value. More generally, by the same topological argument, there are no non-trivial pseudo-deterministic algorithms for quantum states, unitaries, channels, or classical distributions whose domain is the entire set of such objects.

This rules out non-trivial algorithms on connected domains, showing that disconnectedness is a necessary condition for non-triviality. To demonstrate this point, we now introduce a quantum analogue of the \findone problem, \UnifSuppFind,  whose input domain is disconnected.

\begin{definition}[\UnifSuppFind]
    On input $\ket{x}\in \mathcal{S}(\mathbb{C}^{d})$, a quantum state which is promised to be a uniform superposition over some subset of the basis states, the task is to output an index $a \in [d]$ such that $\braket{a|x}\neq 0$. 

    Equivalently, each input $\ket{x}$ has the form $\ket{x}=\frac{1}{\sqrt{|J|}}\sum_{j\in J}\ket{j}$ for some $J\subseteq[d]$, where $J\neq \emptyset$, and the output is some $j_0\in J$.
\end{definition}

Note that the set of these superpositions forms a discrete, and therefore disconnected, subset of $\mathcal{S}(\mathbb{C}^d)$.

\begin{theorem}\label{thm: suppfind}
    The query complexity of \UnifSuppFind is $\Theta(d)$.
\end{theorem}

\begin{proof}
    We first prove the upper bound. This follows from the fact that when the state defined over the set $J\subseteq[d]$ is measured in the computational basis, one gets the smallest element of the support with probability $1/|J|\geq 1/d$. Therefore, one can measure $t=cd$ copies of the state in the computational basis and then output the smallest measured index from these measurements. The probability that the smallest index is measured is lower bounded by the following for $c\geq \ln(3)$.
    \begin{equation*}
        1-(1-1/|J|)^t \geq 1-(1-1/d)^{cd}\geq 1-e^{-c}\geq 2/3
    \end{equation*}

    To prove the lower bound, we consider a pseudo-deterministic algorithm $\alg$ and note that the probability that an algorithm can distinguish between two states is upper bounded by their trace distance. Consider the quantum state $\ket{y}$ whose support size is $d-1$, missing only basis element $k$. Suppose the pseudo-deterministic output of $\alg$ on $\ket{y}$ is $j\neq k$. This state is expressed as 
    \begin{equation*}
        \ket{y}=\sum_{\ell\in[d]\setminus\{k\}}\frac{1}{\sqrt{d-1}}\ket{\ell}.
    \end{equation*}
     Now consider the state
     \begin{equation*}
         \ket z=\sum_{\ell\in[d]\setminus\{j\}}\frac{1}{\sqrt{d-1}}\ket{\ell},
     \end{equation*} which has support size $d-1$ missing only basis element $j$. The pseudo-deterministic output of $\alg$ on $\ket{z}$ cannot be $j$. The trace distance between $t$ copies of $\ket{y}$ and $\ket{z}$ is:
    \[
    \begin{split}
        \frac{1}{2}\|\ket{y}\bra{y}^{\otimes t}-\ket{z}\bra{z}^{\otimes t}\|_1&=\sqrt{1-|\braket{y|z}|^{2t}}\\
        &=\sqrt{1-\left(1-\frac{1}{d-1}\right)^{2t}}\\
        &\leq \sqrt{1-\left(1-\frac{2t}{d-1}\right)}\\
        &=\sqrt{\frac{2t}{d-1}}
    \end{split} 
    \]
    Here, the third line follows by Bernoulli's inequality. For the trace distance to be $\Omega(1)$ so that the two cases are distinguishable, we must have $t=\Omega(d)$. 
\end{proof}

\section*{Acknowledgements}

We thank Rahul Santhanam, Igor Oliveira, and Scott Aaronson for useful discussions on pseudo-determinism.

\bibliographystyle{alphaurl}
\bibliography{ref}

@inproceedings{fishing,
  author       = {Scott Aaronson},
  editor       = {Leonard J. Schulman},
  title        = {{BQP} and the polynomial hierarchy},
  booktitle    = {Proceedings of the 42nd {ACM} Symposium on Theory of Computing, {STOC}
                  2010, Cambridge, Massachusetts, USA, 5-8 June 2010},
  pages        = {141--150},
  publisher    = {{ACM}},
  year         = {2010},
  url          = {https://doi.org/10.1145/1806689.1806711},
  doi          = {10.1145/1806689.1806711},
  bibsource    = {dblp computer science bibliography, https://dblp.org}
}

@inproceedings{GIPS21,
  author       = {Shafi Goldwasser and
                  Russell Impagliazzo and
                  Toniann Pitassi and
                  Rahul Santhanam},
  editor       = {Valentine Kabanets},
  title        = {On the Pseudo-Deterministic Query Complexity of {NP} Search Problems},
  booktitle    = {36th Computational Complexity Conference, {CCC} 2021, July 20-23,
                  2021, Toronto, Ontario, Canada (Virtual Conference)},
  series       = {LIPIcs},
  volume       = {200},
  pages        = {36:1--36:22},
  publisher    = {Schloss Dagstuhl - Leibniz-Zentrum f{\"{u}}r Informatik},
  year         = {2021},
  url          = {https://doi.org/10.4230/LIPIcs.CCC.2021.36},
  doi          = {10.4230/LIPIcs.CCC.2021.36},
  bibsource    = {dblp computer science bibliography, https://dblp.org}
}

@inproceedings{GGR13,
  author       = {Oded Goldreich and
                  Shafi Goldwasser and
                  Dana Ron},
  editor       = {Robert D. Kleinberg},
  title        = {On the possibilities and limitations of pseudodeterministic algorithms},
  booktitle    = {Innovations in Theoretical Computer Science, {ITCS} '13, Berkeley,
                  CA, USA, January 9-12, 2013},
  pages        = {127--138},
  publisher    = {{ACM}},
  year         = {2013},
  url          = {https://doi.org/10.1145/2422436.2422453},
  doi          = {10.1145/2422436.2422453},
  bibsource    = {dblp computer science bibliography, https://dblp.org}
}

@article{Yamakawa2022,
  author       = {Takashi Yamakawa and
                  Mark Zhandry},
  title        = {Verifiable Quantum Advantage without Structure},
  journal      = {J. {ACM}},
  volume       = {71},
  number       = {3},
  pages        = {20},
  year         = {2024},
  url          = {https://doi.org/10.1145/3658665},
  doi          = {10.1145/3658665},
  bibsource    = {dblp computer science bibliography, https://dblp.org}
}

@article{CE05subset,
  author       = {Andrew M. Childs and
                  Jason M. Eisenberg},
  title        = {Quantum algorithms for subset finding},
  journal      = {Quantum Inf. Comput.},
  volume       = {5},
  number       = {7},
  pages        = {593--604},
  year         = {2005},
  url          = {https://doi.org/10.26421/QIC5.7-7},
  doi          = {10.26421/QIC5.7-7},
  bibsource    = {dblp computer science bibliography, https://dblp.org}
}

@article{CDM23query,
  author       = {Arkadev Chattopadhyay and
                  Yogesh Dahiya and
                  Meena Mahajan},
  title        = {Query Complexity of Search Problems},
  journal      = {Electron. Colloquium Comput. Complex.},
  volume       = {{TR23-039}},
  year         = {2023},
  url          = {https://eccc.weizmann.ac.il/report/2023/039},
  eprinttype    = {ECCC},
  eprint       = {TR23-039},
  bibsource    = {dblp computer science bibliography, https://dblp.org}
}

@inproceedings{Szeg04,
  author       = {Mario Szegedy},
  title        = {Quantum Speed-Up of Markov Chain Based Algorithms},
  booktitle    = {45th Symposium on Foundations of Computer Science {(FOCS} 2004), 17-19
                  October 2004, Rome, Italy, Proceedings},
  pages        = {32--41},
  publisher    = {{IEEE} Computer Society},
  year         = {2004},
  url          = {https://doi.org/10.1109/FOCS.2004.53},
  doi          = {10.1109/FOCS.2004.53},
  bibsource    = {dblp computer science bibliography, https://dblp.org}
}

@inproceedings{Bel12,
  author       = {Aleksandrs Belovs},
  title        = {Learning-Graph-Based Quantum Algorithm for k-Distinctness},
  booktitle    = {53rd Annual {IEEE} Symposium on Foundations of Computer Science, {FOCS}
                  2012, New Brunswick, NJ, USA, October 20-23, 2012},
  pages        = {207--216},
  publisher    = {{IEEE} Computer Society},
  year         = {2012},
  url          = {https://doi.org/10.1109/FOCS.2012.18},
  doi          = {10.1109/FOCS.2012.18},
  bibsource    = {dblp computer science bibliography, https://dblp.org}
}

@inproceedings{GGH18,
  author       = {Shafi Goldwasser and
                  Ofer Grossman and
                  Dhiraj Holden},
  editor       = {Anna R. Karlin},
  title        = {Pseudo-Deterministic Proofs},
  booktitle    = {9th Innovations in Theoretical Computer Science Conference, {ITCS}
                  2018, January 11-14, 2018, Cambridge, MA, {USA}},
  series       = {LIPIcs},
  volume       = {94},
  pages        = {17:1--17:18},
  publisher    = {Schloss Dagstuhl - Leibniz-Zentrum f{\"{u}}r Informatik},
  year         = {2018},
  url          = {https://doi.org/10.4230/LIPIcs.ITCS.2018.17},
  doi          = {10.4230/LIPICS.ITCS.2018.17},
  bibsource    = {dblp computer science bibliography, https://dblp.org}
}

@inproceedings{GGMW20,
  author       = {Shafi Goldwasser and
                  Ofer Grossman and
                  Sidhanth Mohanty and
                  David P. Woodruff},
  editor       = {Thomas Vidick},
  title        = {Pseudo-Deterministic Streaming},
  booktitle    = {11th Innovations in Theoretical Computer Science Conference, {ITCS}
                  2020, January 12-14, 2020, Seattle, Washington, {USA}},
  series       = {LIPIcs},
  volume       = {151},
  pages        = {79:1--79:25},
  publisher    = {Schloss Dagstuhl - Leibniz-Zentrum f{\"{u}}r Informatik},
  year         = {2020},
  url          = {https://doi.org/10.4230/LIPIcs.ITCS.2020.79},
  doi          = {10.4230/LIPICS.ITCS.2020.79},
  bibsource    = {dblp computer science bibliography, https://dblp.org}
}

@inproceedings{BKKS23,
  author       = {Vladimir Braverman and
                  Robert Krauthgamer and
                  Aditya Krishnan and
                  Shay Sapir},
  editor       = {Kousha Etessami and
                  Uriel Feige and
                  Gabriele Puppis},
  title        = {Lower Bounds for Pseudo-Deterministic Counting in a Stream},
  booktitle    = {50th International Colloquium on Automata, Languages, and Programming,
                  {ICALP} 2023, July 10-14, 2023, Paderborn, Germany},
  series       = {LIPIcs},
  volume       = {261},
  pages        = {30:1--30:14},
  publisher    = {Schloss Dagstuhl - Leibniz-Zentrum f{\"{u}}r Informatik},
  year         = {2023},
  url          = {https://doi.org/10.4230/LIPIcs.ICALP.2023.30},
  doi          = {10.4230/LIPICS.ICALP.2023.30},
  bibsource    = {dblp computer science bibliography, https://dblp.org}
}

@inproceedings{GGS23,
  author       = {Ofer Grossman and
                  Meghal Gupta and
                  Mark Sellke},
  title        = {Tight Space Lower Bound for Pseudo-Deterministic Approximate Counting},
  booktitle    = {64th {IEEE} Annual Symposium on Foundations of Computer Science, {FOCS}
                  2023, Santa Cruz, CA, USA, November 6-9, 2023},
  pages        = {1496--1504},
  publisher    = {{IEEE}},
  year         = {2023},
  url          = {https://doi.org/10.1109/FOCS57990.2023.00091},
  doi          = {10.1109/FOCS57990.2023.00091},
  bibsource    = {dblp computer science bibliography, https://dblp.org}
}

@article{GGH19,
  author       = {Michel X. Goemans and
                  Shafi Goldwasser and
                  Dhiraj Holden},
  title        = {Doubly-Efficient Pseudo-Deterministic Proofs},
  journal      = {CoRR},
  volume       = {abs/1910.00994},
  year         = {2019},
  url          = {http://arxiv.org/abs/1910.00994},
  eprinttype    = {arXiv},
  eprint       = {1910.00994},
  bibsource    = {dblp computer science bibliography, https://dblp.org}
}

@inproceedings{CLORS23,
  author       = {Lijie Chen and
                  Zhenjian Lu and
                  Igor C. Oliveira and
                  Hanlin Ren and
                  Rahul Santhanam},
  title        = {Polynomial-Time Pseudodeterministic Construction of Primes},
  booktitle    = {64th {IEEE} Annual Symposium on Foundations of Computer Science, {FOCS}
                  2023, Santa Cruz, CA, USA, November 6-9, 2023},
  pages        = {1261--1270},
  publisher    = {{IEEE}},
  year         = {2023},
  url          = {https://doi.org/10.1109/FOCS57990.2023.00074},
  doi          = {10.1109/FOCS57990.2023.00074},
  bibsource    = {dblp computer science bibliography, https://dblp.org}
}

@article{GG11,
  author       = {Eran Gat and
                  Shafi Goldwasser},
  title        = {Probabilistic Search Algorithms with Unique Answers and Their Cryptographic
                  Applications},
  journal      = {Electron. Colloquium Comput. Complex.},
  volume       = {{TR11-136}},
  year         = {2011},
  url          = {https://eccc.weizmann.ac.il/report/2011/136},
  eprinttype    = {ECCC},
  eprint       = {TR11-136},
  bibsource    = {dblp computer science bibliography, https://dblp.org}
}

@inproceedings{DPWV22,
  author       = {Peter Dixon and
                  Aduri Pavan and
                  Jason Vander Woude and
                  N. V. Vinodchandran},
  editor       = {Stefano Leonardi and
                  Anupam Gupta},
  title        = {Pseudodeterminism: promises and lowerbounds},
  booktitle    = {{STOC} '22: 54th Annual {ACM} {SIGACT} Symposium on Theory of Computing,
                  Rome, Italy, June 20 - 24, 2022},
  pages        = {1552--1565},
  publisher    = {{ACM}},
  year         = {2022},
  url          = {https://doi.org/10.1145/3519935.3520043},
  doi          = {10.1145/3519935.3520043},
  bibsource    = {dblp computer science bibliography, https://dblp.org}
}

@article{DPV21,
  author       = {Peter Dixon and
                  Aduri Pavan and
                  N. V. Vinodchandran},
  title        = {Promise Problems Meet Pseudodeterminism},
  journal      = {CoRR},
  volume       = {abs/2103.08589},
  year         = {2021},
  url          = {https://arxiv.org/abs/2103.08589},
  eprinttype    = {arXiv},
  eprint       = {2103.08589},
  bibsource    = {dblp computer science bibliography, https://dblp.org}
}

@inproceedings{ABKRT21,
  author       = {Scott Aaronson and
                  Shalev Ben{-}David and
                  Robin Kothari and
                  Shravas Rao and
                  Avishay Tal},
  editor       = {Samir Khuller and
                  Virginia Vassilevska Williams},
  title        = {Degree vs. approximate degree and Quantum implications of Huang's
                  sensitivity theorem},
  booktitle    = {{STOC} '21: 53rd Annual {ACM} {SIGACT} Symposium on Theory of Computing,
                  Virtual Event, Italy, June 21-25, 2021},
  pages        = {1330--1342},
  publisher    = {{ACM}},
  year         = {2021},
  url          = {https://doi.org/10.1145/3406325.3451047},
  doi          = {10.1145/3406325.3451047},
  bibsource    = {dblp computer science bibliography, https://dblp.org}
}

@book{NC16,
  author       = {Michael A. Nielsen and
                  Isaac L. Chuang},
  title        = {Quantum Computation and Quantum Information (10th Anniversary edition)},
  publisher    = {Cambridge University Press},
  year         = {2016},
  url          = {https://www.cambridge.org/de/academic/subjects/physics/quantum-physics-quantum-information-and-quantum-computation/quantum-computation-and-quantum-information-10th-anniversary-edition?format=HB},
  isbn         = {978-1-10-700217-3},
  bibsource    = {dblp computer science bibliography, https://dblp.org}
}

@article{bassirian2021certified,
  title={On certified randomness from Fourier sampling or random circuit sampling},
  author={Bassirian, Roozbeh and Bouland, Adam and Fefferman, Bill and Gunn, Sam and Tal, Avishay},
  journal={arXiv preprint arXiv:2111.14846},
  year={2021}
}

@article{Simon97,
  author       = {Daniel R. Simon},
  title        = {On the Power of Quantum Computation},
  journal      = {{SIAM} J. Comput.},
  volume       = {26},
  number       = {5},
  pages        = {1474--1483},
  year         = {1997},
  url          = {https://doi.org/10.1137/S0097539796298637},
  doi          = {10.1137/S0097539796298637},
  bibsource    = {dblp computer science bibliography, https://dblp.org}
}

@article{BBCMW01,
  author       = {Robert Beals and
                  Harry Buhrman and
                  Richard Cleve and
                  Michele Mosca and
                  Ronald de Wolf},
  title        = {Quantum lower bounds by polynomials},
  journal      = {J. {ACM}},
  volume       = {48},
  number       = {4},
  pages        = {778--797},
  year         = {2001},
  url          = {https://doi.org/10.1145/502090.502097},
  doi          = {10.1145/502090.502097},
  bibsource    = {dblp computer science bibliography, https://dblp.org}
}

@article{farhi1998limit,
  title={Limit on the speed of quantum computation in determining parity},
  author={Farhi, Edward and Goldstone, Jeffrey and Gutmann, Sam and Sipser, Michael},
  journal={Physical Review Letters},
  volume={81},
  number={24},
  pages={5442},
  year={1998},
  publisher={APS}
}

@inproceedings{Gall14,
  author       = {Fran{\c{c}}ois Le Gall},
  title        = {Improved Quantum Algorithm for Triangle Finding via Combinatorial
                  Arguments},
  booktitle    = {55th {IEEE} Annual Symposium on Foundations of Computer Science, {FOCS}
                  2014, Philadelphia, PA, USA, October 18-21, 2014},
  pages        = {216--225},
  publisher    = {{IEEE} Computer Society},
  year         = {2014},
  url          = {https://doi.org/10.1109/FOCS.2014.31},
  doi          = {10.1109/FOCS.2014.31},
  bibsource    = {dblp computer science bibliography, https://dblp.org}
}

@inproceedings{AH23,
  author       = {Scott Aaronson and
                  Shih{-}Han Hung},
  editor       = {Barna Saha and
                  Rocco A. Servedio},
  title        = {Certified Randomness from Quantum Supremacy},
  booktitle    = {Proceedings of the 55th Annual {ACM} Symposium on Theory of Computing,
                  {STOC} 2023, Orlando, FL, USA, June 20-23, 2023},
  pages        = {933--944},
  publisher    = {{ACM}},
  year         = {2023},
  url          = {https://doi.org/10.1145/3564246.3585145},
  doi          = {10.1145/3564246.3585145},
  bibsource    = {dblp computer science bibliography, https://dblp.org}
}

@article{GraphColl,
  author       = {Andris Ambainis and
                  Kaspars Balodis and
                  Janis Iraids and
                  Raitis Ozols and
                  Juris Smotrovs},
  title        = {Parameterized Quantum Query Complexity of Graph Collision},
  journal      = {CoRR},
  volume       = {abs/1305.1021},
  year         = {2013},
  url          = {http://arxiv.org/abs/1305.1021},
  eprinttype    = {arXiv},
  eprint       = {1305.1021},
  bibsource    = {dblp computer science bibliography, https://dblp.org}
}

@inproceedings{AGJK20walk,
  author       = {Andris Ambainis and
                  Andr{\'{a}}s Gily{\'{e}}n and
                  Stacey Jeffery and
                  Martins Kokainis},
  editor       = {Konstantin Makarychev and
                  Yury Makarychev and
                  Madhur Tulsiani and
                  Gautam Kamath and
                  Julia Chuzhoy},
  title        = {Quadratic speedup for finding marked vertices by quantum walks},
  booktitle    = {Proceedings of the 52nd Annual {ACM} {SIGACT} Symposium on Theory
                  of Computing, {STOC} 2020, Chicago, IL, USA, June 22-26, 2020},
  pages        = {412--424},
  publisher    = {{ACM}},
  year         = {2020},
  url          = {https://doi.org/10.1145/3357713.3384252},
  doi          = {10.1145/3357713.3384252},
  bibsource    = {dblp computer science bibliography, https://dblp.org}
}

\appendix

\newpage

\section{Proof of \cref{thm: wb PsQ to FPBQP}}\label{App:WhiteBox}

\WhiteBox*

\begin{proof}
    For the ``if'' direction, we assume that such a $\P$-reduction exists. Given an input to the search problem $\probR$, we use the reduction as an algorithm and run the $\BQP$ algorithm on the output. The solution is pseudo-deterministic as, on a single input to $\probR:X$, the reduction will reduce to the same instances of the $\BQP$ decision problem as any such reduction is deterministic, let this reduce to the series of inputs to the $\BQP$ decision problem $(Y_{j})_{j}$. As the $\BQP$ problem is a decision problem, the solution on inputs $(Y_{j})_j$ will be the same on multiple repetitions of this algorithm for $\probR$.  

    For the other direction, suppose the algorithm $\mathcal{A}$ is a pseudo-deterministic quantum algorithm for this search problem. Let us define the set $\probR'$ in the following way,
    \begin{equation*}
        (x,i,b)\in \probR'\iff \mathbb{P}[\mathcal{A}(x)_{i}=b]>\frac{2}{3}.
    \end{equation*}
    By definition, the set $\probR'$ is decidable in $\BQP$, and thus solving $\probR$ is  $\P$-reducible to deciding  $\probR'$.
\end{proof}

\section{Proof of \cref{thm:Find1QComplete}}\label{sec:find1QComplete}

\FindoneQComplete*

\begin{proof}
    We build a reduction from $R$ to \findone.
    Let $\mathcal{V}$ be a classical deterministic algorithm for verifying any solution to $\probR$ with query complexity $ \V(R)$, and $\mathcal{R}$ be a classical randomized algorithm that solves $\probR$ with query complexity $\R(R)$.

    W.l.o.g, $\mathcal{R}$ is a uniform distribution over deterministic algorithms with multiplicity, where each deterministic algorithm in the support of $\mathcal{R}$ makes at most $\R(R)$ queries. 
    Additionally, for any input $X$, the probability that a random algorithm drawn from $\mathcal{R}$ succeeds on input $X$ is at least $2/3$. 

    If we randomly draw $100N$ of deterministic algorithms from the support of $\mathcal{R}$, which corresponds to a new randomized algorithm $\mathcal{R}'$.
    By the Chernoff bound, the probability that for any given $X$, at least $0.49$ fraction of the algorithms in $\mathcal{R}'$ fail is bounded above by $2e^{-N}$. 
Thus, by taking a union bound over all possible inputs $X\in\{0,1\}^{N}$, the probability that a randomly chosen $\mathcal{R}'$ is also a valid randomized algorithm for $\probR$ is $1 - o(1)$. Moreover, the support size of $\mathcal{R}'$ is $100N$.
    
    We now solve $\probR$ with input $X \in \{0,1\}^N$ by reducing it to a \findone instance on input string $Z\in \{0,1\}^{100N}$. Specifically, for any $i \in [100N]$, $Z_i$ is $1$ if the $i$-th deterministic algorithm in the support of $\mathcal{R}'$ produces a correct solution to $X$. Note that $Z_i$ can be calculated by $(\R(R)+\V(R))$ queries to string $X$. Moreover, any solution to the \findone instance $Z$ leads to a certain solution for $R$ on instance $X$.

    Finally, given the fact $\psQ(\findone) = \Theta(\sqrt{N})$~\cite{GIPS21}, we get a pseudo-deterministic quantum algorithm for $R$ with $O\big((\R(R)+\V(R))\cdot \sqrt{N}\big)$ queries.

\end{proof}

\end{document}